\documentclass[11pt]{amsart}

\pdfoutput = 1

\usepackage[english]{babel}
\usepackage[autostyle]{csquotes}
\usepackage{amsmath,amssymb,amsthm}
\usepackage{amsfonts}
\usepackage{amsxtra}

\usepackage{array,dcolumn}
\usepackage{graphicx}
\usepackage[hyperfootnotes=true,
        pdffitwindow=true,
        plainpages=false,
        pdfpagelabels=true,
        pdfpagemode=UseOutlines,
        pdfpagelayout=SinglePage,
                hyperindex,]{hyperref}

\usepackage{lmodern}
\usepackage[T1]{fontenc}
\usepackage[utf8]{inputenc}
\usepackage[activate={true,nocompatibility},spacing,kerning]{microtype}
\usepackage{bm}
\usepackage{bbm}
\usepackage[sc,osf]{mathpazo}
\microtypecontext{spacing=nonfrench}

  \calclayout

\newcommand{\mathsym}[1]{{}}
\newcommand{\unicode}[1]{{}}

\theoremstyle{plain}
\newtheorem{theorem}{Theorem}

\newtheorem{corollary}[theorem]{Corollary}
\newtheorem{proposition}[theorem]{Proposition}

\theoremstyle{definition}

\theoremstyle{remark}
\newtheorem{remark}[theorem]{Remark}

\newcommand{\+}{\!+\!}

\numberwithin{equation}{section}

\begin{document}


\title[Random matrix ensembles and the structure function]{Differential identities for the structure function of some random matrix ensembles}
\author{Peter J. Forrester}
\address{School of Mathematics and Statistics, 
ARC Centre of Excellence for Mathematical \& Statistical Frontiers,
University of Melbourne, Victoria 3010, Australia}
\email{pjforr@unimelb.edu.au}

\date{\today}


\begin{abstract}
The structure function of a random matrix ensemble can be specified in terms of the covariance of the linear
statistics $\sum_{j=1}^N e^{i k_1 \lambda_j}$, $\sum_{j=1}^N e^{-i k_2 \lambda_j}$ for Hermitian matrices,
and the same with the eigenvalues $\lambda_j$ replaced by the eigenangles $\theta_j$ for unitary matrices. As such it can be written in terms of
the Fourier transform of the density-density correlation $\rho_{(2)}$. For the circular $\beta$-ensemble of unitary matrices,
and with $\beta$ even, we characterise the bulk scaling limit of $\rho_{(2)}$ as the solution of a linear differential equation
of order $\beta + 1$ --- a duality relates $\rho_{(2)}$ with $\beta$ replaced by $4/\beta$ to the same equation. Asymptotics 
obtained in the case $\beta = 6$ from this characterisation are combined with previously established results to determine the explicit form of the
degree 10 palindromic polynomial in $\beta/2$ which determines the coefficient of $|k|^{11}$ in the small $|k|$ expansion of the structure function for general 
$\beta > 0$. For the Gaussian unitary ensemble  we give a reworking of a recent derivation and generalisation,
due to  Okuyama, of an identity relating the structure function to simpler quantities in the Laguerre unitary ensemble first
derived in random matrix theory by Br\'ezin and Hikami. This is used to determine various scaling limits, many of
which relate to the dip-ramp-plateau effect emphasised in recent studies of many body quantum chaos, and allows too
for rates of convergence to be established.

\end{abstract}


\maketitle


\section{Introduction}\label{S1}
\subsection{Context of the structure function --- variance}\label{S1.1}
A task of random matrix theory is to provide precise analytic statements regarding model
spectra for purposes of quantification and applications. The feature of the spectra to be studied in
this work is the structure function, calculated from the Fourier transform of the (truncated
or connected) density-density
correlation, and the model spectra considered are from the Gaussian unitary ensemble (GUE) of complex
Hermitian matrices, and the
circular $\beta$-ensemble of random unitary matrices, which includes as a special case the Haar
distribution on the matrix group $U(N)$.

With $\langle \cdot \rangle$ denoting the ensemble average, 
and $\{ \lambda_j \}_{j=1}^N$ the eigenvalues,
the density-density correlation
$N_{(2)}(\lambda, \lambda')$ is defined by
\begin{equation}\label{1}
N_{(2)}(\lambda , \lambda') = \langle n_{(1)}(\lambda)  n_{(1)}(\lambda') \rangle -
\langle n_{(1)}(\lambda )   \rangle   \langle n_{(1)}(\lambda')   \rangle, 
\end{equation}
where $n_{(1)}(\lambda)$ denotes the microscopic eigenvalue density
\begin{equation}\label{2}
n_{(1)}(\lambda ) = \sum_{j=1}^N \delta (\lambda - \lambda_j).
\end{equation}
Note the decomposition
\begin{align}\label{2a}
 \langle n_{(1)}(\lambda)  n_{(1)}(\lambda') \rangle & = \Big \langle
 \sum_{j,k=1, j \ne k}^N   \delta (\lambda - \lambda_j)  \delta (\lambda' - \lambda_k)  \Big \rangle + \delta (\lambda - \lambda') 
 \Big \langle \sum_{j=1}^N \delta ( \lambda - \lambda_j) \Big \rangle \nonumber \\
 & = \rho_{(2)}(\lambda, \lambda') + \delta(\lambda - \lambda') \rho_{(1)}(\lambda),
 \end{align}
 where $ \rho_{(2)}(\lambda, \lambda')$ denotes the two-point correlation function, and $\rho_{(1)}(\lambda)$ the
 spectral density.
The correlation function (\ref{1}) fully determines the variance of a linear statistic
$A = \sum_{j=1}^N a(\lambda_j)$ of the eigenvalues. Thus
(see e.g.~\cite[Prop.~14.3.2]{Fo10})
\begin{equation}\label{3}
{\rm Var} \, (A) := \langle A^2 \rangle = \int_I d \lambda \int_I d \lambda' \,
 a(\lambda) a(\lambda') N_{(2)}(\lambda , \lambda'),
 \end{equation}
 where $I$ denotes the eigenvalue support.
 
 Generically $N_{(2)}(\lambda , \lambda')$ is a function of two variables, but
 in a translationly invariant state simplifies to depend on only one,
 $N_{(2)}(\lambda , \lambda') = N_{(2)}(\lambda - \lambda',0)$. 
 One setting in which this happens is when the matrices are unitary,
 so the eigenvalues are on the unit circle, and furthermore the eigenvalue density is
 uniform, as is the case for the circular ensembles.
  Another is the bulk scaling limit, when the eigenvalues have been scaled $\lambda =
 s_N x$, so that in the new variables $x$, and after the limit $N \to \infty$
 has been taken, the eigenvalue density is  a fixed
 constant which we will take to be unity. For example, in the GUE 
 this happens with the choice $s_N = \pi/ \sqrt{2N}$
 (see e.g.~\cite[Prop.~7.1.1]{Fo10}), while in the circular ensembles, the
 eigenvalues are first parametrised by their angle $\theta$, then scaling $\theta = s_N x$
 with $s_N = 2 \pi / N$. In this setting
 \begin{equation}\label{4}
 N_{(2), \infty}(x , x') := \lim_{N \to \infty} s_N^2 N_{(2)}(s_N x , s_N x')
  \end{equation}
  is well defined and is translationally invariant $N_{(2), \infty}(x , x') = N_{(2), \infty}(x-x',0)$.
 
The bulk scaled structure function $S(k)$ is defined as the Fourier transform of the scaled
density-density correlation
 \begin{equation}\label{5}
 S(k) := \int_{-\infty}^\infty  N_{(2), \infty}(x , 0) e^{i k x} \, dx.
\end{equation} 
Now replace the linear statistic $A$ by the corresponding scaled quantity
 \begin{equation}\label{5a}
\tilde{A} = \sum_{j=1}^N  a( \lambda_j/ s_N) = \sum_{j=1}^N a(x_j).
\end{equation} 
Changing
variables in this modification of (\ref{3}),
we see by taking the inverse
transform of (\ref{5}) to substitute for the translationally invariant quantity
$N_{(2), \infty}(x , 0)$ that
 \begin{equation}\label{6}
  \lim_{N \to \infty} {\rm Var} \, ( \tilde{A}) = {1 \over 2 \pi} \int_{-\infty}^\infty | \hat{a}(k) |^2 S(k) \, dk, \quad
  \hat{a}(k) :=  \int_{-\infty}^\infty a(x) e^{i k x} \, dx.
  \end{equation}
  
  Universality (see e.g.~\cite{EY17}) tells us that the bulk scaling limit of $ S(k)$ is the
  same for both the GUE and the circular unitary ensemble (CUE) of Haar distributed
  unitary matrices. For these ensembles, labelled by the Dyson index $\beta = 2$,
   \begin{equation}\label{7}
  S(k) \Big |_{\beta = 2} =  \int_{-\infty}^\infty \Big (  - \Big ( {\sin \pi x \over \pi x} \Big )^2  + \delta(x) \Big )  e^{i k x}  \, dx = \left \{ \begin{array}{ll} \displaystyle {  |k| \over 2 \pi}, & |k| < 2 \pi \\[1ex]
  1, & |k| \ge 2 \pi. \end{array} \right.
  \end{equation}
  We remark that the right hand side of (\ref{6}) with the substitution (\ref{7}), as obtained by
  Montgomery \cite{Mo73} in the context of the statistical properties of the zeros of the
  Riemann zeta functions, was what lead Dyson to observe a relationship with GUE (or
  equivalently CUE) eigenvalue statistics; see the account in \cite{Bo05}.
  
  Suppose that associated with the function $a(x)$ defining the linear statistic is a
  length scale $L$, and in fact $a(x) = c(x/L)$ for some function $c(y)$. In this
  circumstance write $\tilde{A} = \tilde{A}_L$. Then a change of variables in (\ref{6}) shows
    \begin{equation}\label{8} 
  \lim_{N \to \infty}   {\rm Var} \, ( \tilde{A}_L) = {L \over 2 \pi} \int_{-\infty}^\infty | \hat{c}(k) |^2 S(k/L) \, dk.
 \end{equation}   
 It had been observed by Dyson (and Mehta) \cite{DM63}, a decade prior to the exchange with
 Montgomery, that the small $k$ form exhibited in  (\ref{7}) implies that  the
 variance formula (\ref{8}) for the bulk scaled GUE has the well defined $L \to \infty$ limit
   \begin{equation}\label{8y} 
  \lim_{L \to \infty}  \lim_{N \to \infty}  {\rm Var} \, ( \tilde{A}_L) \Big |_{\beta = 2} = {1 \over (2 \pi)^2} \int_{-\infty}^\infty | \hat{c}(k) |^2 |k| \, dk,
 \end{equation}  
   assuming $\hat{c}(k) $ decays sufficiently fast. A very similar formula holds true for other symmetry classes,
   distinguished by the Dyson parameter $\beta$. Thus for the bulk scaled limit of the circular
   $\beta$-ensemble, upon the knowledge that (see e.g.~\cite[Eq.~(14.8)]{Fo10})
  \begin{equation}\label{9}   
  S(k) \mathop{\sim}\limits_{|k| \to 0} {|k| \over \pi \beta},
  \end{equation}
  one sees from (\ref{8y}) 
  \begin{equation}\label{10}    
    \lim_{L \to \infty}   \lim_{N \to \infty}  {\rm Var} \, ( \tilde{A}_L)  = {1 \over2 \pi^2  \beta } \int_{-\infty}^\infty | \hat{c}(k) |^2 |k| \, dk.
  \end{equation}
  
  The limit formula (\ref{10}) can be extended to an asymptotic formula. Thus 
  if we set 
\begin{equation}\label{tR2}
f(k;\beta) = {\pi \beta \over |k|} S(k;\beta), \qquad 0 < k < {\rm min} \, (2 \pi, \pi \beta),
\end{equation}
and define $f$ for $k < 0$ by analytic continuation then we know from \cite{FJM00} that
$f(k;\beta)$ is an analytic function of $k$ for $|k| < {\rm min} \, (2 \pi, \pi \beta)$. Writing $S(k) = S(k;\beta)$
and expanding
$S(k;\beta) = \sum_{p=1}^\infty \alpha_p(\beta) |k|^p$ gives the sought extension of (\ref{10}) 
\begin{equation}\label{11} 
  \lim_{N \to \infty} {\rm Var} \, ( \tilde{A}_L) \mathop{\sim}\limits_{L \to \infty}
  \sum_{p=1}^\infty {\alpha_p(\beta) \over L^{p-1}}  \int_{-\infty}^\infty | \hat{c}(k) |^2 |k|^p \, dk.
  \end{equation}
  
  In Section \ref{S2} below we show that consideration of
 integrability features of the two-point
 correlation function $\rho_{(2)}(x,x')$, by way of a particular differential-difference system
 from the theory of the Selberg integral, has consequence for the expansion
  coefficients in (\ref{11}). Generally, integrability is a major theme in random matrix theory,
  with perhaps its  best known manifestation being nonlinear differential equations associated with gap
  probabilities; see e.g.~\cite[Ch.~8 \& 9]{Fo10}.
  In Section \ref{S2} we take up the task of characterising
  \begin{equation}\label{11a}
  \rho_{(2), \infty}(x , x') := \lim_{N \to \infty} s_N^2 \rho_{(2)}(s_N x , s_N x'),
  \end{equation}
  which according to (\ref{2a}) is equal to $ N_{(2), \infty}(x , x') - 1$ for $x \ne x'$, in terms of
  homogeneous {\it linear} differential equations. We do this for the circular $\beta$-ensemble with
  the values $\beta = 2,4,6, 8$ --- see Propositions \ref{p2}, \ref{p4} and Remark \ref{R11a}, with the differential equations
  being of degree $\beta + 1$ ---
  and $\beta = 1,2/3$ --- see Remarks \ref{R3} and \ref{R6a}.
  As an application, we extend knowledge of $\{ \alpha_p(\beta) \}_{p=1}^{10}$ --- which after
  appropriate normalisation are known to be palindromic polynomials in $\beta / 2$ of degree $p-1$ --- in (\ref{11}) from the
  earlier works \cite{FJM00,WF14} to now include $ \alpha_{11}(\beta)$.
  
    \begin{proposition}\label{p5}
Let  $[ k^{10}] h(k)$ denote the coefficient of $k^{10}$ in $h(k)$, and let $x = \beta/2$. We have
 \begin{multline}\label{f9}   
[ k^{10}] {\beta \pi \over |k|} S(k;\beta) = 
\Big ( {1 \over 2 \pi x} \Big )^{10} (x - 1)^2 \\
\times \Big ( x^8-\frac{1523 x^7}{420}+\frac{2529 x^6}{350}-\frac{256189 x^5}{25200}+
\frac{142463 x^4}{12600}-\frac{256189 x^3}{25200}+\frac{2529 x^2}{350}-\frac{1523 x}{420}+1 \Big ).
\end{multline}
\end{proposition}
  
 \subsection{Context of the structure function --- dip, ramp, plateau}\label{s1.2}  
 It has been commented in the paragraph containing (\ref{4}) that the circular ensembles are translationally
 invariant for finite $N$. This greatly simplifies the analysis of the density-density correlation, which
 otherwise is a function of two variables, and enables our computations of \S \ref{S2} below.
 Universality tells us that the bulk scaling state of the Gaussian $\beta$-ensemble and circular $\beta$-ensemble are the
 same. Starting from finite $N$, to access this limiting state consideration of the circular ensemble is
 simpler. See the recent works \cite{ROK17, ROK20} on the exact calculation of the bulk scaled
 power spectrum for the GUE or CUE for a further example of this strategy.
 
 There are other classes of problems in random matrix theory where the probe is not the local scale of the bulk
 eigenvalues, but rather a global scale where the entirety of the spectrum plays a role. An example is the 
 variance of the linear statistic in the GUE with the scale $s_N = \sqrt{2N}$ in (\ref{5a}) --- the significance
 of this choice is that $|\lambda| < s_N$ is the leading order support of the spectrum. The formula (\ref{6}) with
 $S(k)$ given by (\ref{8}) is now replaced by (see e.g.~\cite{PS11})
   \begin{equation}\label{8a} 
   \lim_{N \to \infty} {\rm Var} \, ( \tilde{A}) =  {1 \over 4 \pi^2} \int_{-1}^1 dx   \int_{-1}^1 dy \, \Big ( {a(x) - a(y) \over x - y } \Big )^2
  {1 - xy \over \sqrt{1 - x^2} \sqrt{1 - y^2}}.
  \end{equation}
  
  Another example where the entire spectrum plays a role is when the density-density correlation is
  averaged,
  $$
  \int_{-\infty}^\infty N_{(2)}(\lambda, \lambda + s) \, d \lambda =: \bar{N}_{(2)}(s),
  $$
  so obtaining a function of a single variable. Notice that
    \begin{align}\label{8b} 
    \bar{S}_N(k) & :=   \int_{-\infty}^\infty  \bar{N}_{(2)}(s) e^{i k s} \, ds =
    \int_{-\infty}^\infty d \lambda \, e^{-i k \lambda}   \int_{-\infty}^\infty d \lambda' \, e^{i k \lambda'} 
  N_{(2)}(\lambda, \lambda') \nonumber \\
  &
  = \Big \langle \Big | \sum_{j=1}^N e^{i k \lambda_j} \Big |^2 \Big \rangle - \bigg |
   \Big \langle\sum_{j=1}^N e^{i k \lambda_j}  \Big \rangle \bigg |^2.
  \end{align}   
  In the case of the GUE  the first of the averages on the RHS of (\ref{8b}),
  as a stand alone quantity, has received
 recent attention in the context of the scrambling of information in
  black holes \cite{C+17,CMS17}, and 
  as a probe of many body quantum chaos
  \cite{CHLY17,TGS18, CMC19,CH19}. Prominent in these studies is the graphical shape, exhibiting a dip, ramp and
  then plateau, referred to as a ``correlation hole'' in earlier literature \cite{LLJP86}.
  
  Many years before the prominence given to $ \bar{S}_N(k) $, and associated averages, 
   for the GUE
  by these new fields of application, Br\'ezin and Hikami \cite{BH97} had obtained the striking identity
   \begin{equation}\label{bhX}
     \bar{S}_N^{(G)}(k) =
      \int_0^k  t K_N^{(L)}(t^2/2,t^2/2) \Big |_{a=0} \, dt.
  \end{equation}
  Here and below the use of the superscript $(G)$ on the LHS indicates the quantity is with respect to the
  GUE, while on the RHS $K_N^{(L)}$ denotes the correlation kernel for the Laguerre unitary ensemble (LUE) ---
  see (\ref{kp1x}) and  (\ref{kp1y}) below for its specification. It is
  clear  upon differentiating both sides that differential identities --- the theme of Section \ref{S1.1} --- play a role in the underlying
  theory. Inspired  by the recent derivation of (\ref{bh}) given by Okuyama \cite{Ok19}, we
  bring these structures to the fore.  One consequence is a generalisation of (\ref{bhX}), relating to a particular
  covariance, defined for general  linear statistics $A,B$, by
   \begin{equation}\label{Cov}
  {\rm Cov} \,(A,B) = \langle (A - \langle A \rangle) (B - \langle B \rangle) \rangle.
    \end{equation} 
  
  \begin{proposition}\label{p11}
 Define
  \begin{equation}\label{H}
  H^{(L)}(t_1,t_2) = {t_1 + t_2 \over 2} K_N^{(L)}(t_1^2/2, t_2^2/2) \Big |_{a=0}.
  \end{equation}
  We have
   \begin{equation}\label{H1X} 
   {\rm Cov } \, \Big ( \sum_{j=1}^N e^{i k_1 \lambda_j},   \sum_{j=1}^N e^{-i k_2 \lambda_j} \Big )^{(G)}
   =
   \int_0^{k_2} H^{(L)}(k_1 - k_2 +s,  s) \, ds.
  \end{equation}
  \end{proposition}
  
  We show how the formulas (\ref{bhX}) and (\ref{H1X}) can be used in the derivation of several limit
  theorems, and moreover in some cases allow for precise determination of the rate of convergence.

\section{Differential equations for the bulk two-point correlation function}\label{S2}

\subsection{The method}\label{s2.1}
The circular $\beta$-ensemble is specified by the eigenvalue probability density function
\begin{equation}\label{C1}
p^{\rm C}(\theta_1,\dots,\theta_N) := {1 \over \mathcal N_{N,\beta}^{\rm C}} \prod_{1 \le j < k \le N} | e^{i \theta_k} - e^{i \theta_j} |^\beta,
\end{equation}
supported on $0 \le \theta_l < 2 \pi$ $(l=1,\dots,N)$. Here $ \mathcal N_{N,\beta}^{\rm C} = (2 \pi)^N
\Gamma(\beta N/2 + 1)/ (\Gamma(\beta/2 + 1))^N$ is the normalisation; see e.g.~\cite[Prop.~4.7.2]{Fo10}
for its derivation. Let $\rho_{(2),N}^{\rm C}(\theta_1,\theta_2)$ denote the corresponding two-point correlation function.
This is defined by integrating over each of $\theta_3,\dots, \theta_N$, and multiplying by $N(N-1)$ as a normalisation.
Thus
\begin{equation}\label{C1a}
\rho_{(2),N}^{\rm C}(\theta_1,\theta_2) = {N(N-1) \over N_{N,\beta}^{\rm C}} \int_0^{2 \pi} d \theta_3 \cdots \int_0^{2 \pi} d \theta_N \,
p^{\rm C}(\theta_1,\dots,\theta_N).
\end{equation}
Due to (\ref{C1}) being invariant under rotation, i.e.~the replacement $\theta_l \mapsto \theta_l + \phi$, $(l=1,\dots,N)$, one sees that
\begin{equation}\label{C1b}
\rho_{(2),N}^{\rm C}(\theta_1,\theta_2) = \rho_{(2),N}^{\rm C}(\theta_1 - \theta_2, 0),
 \end{equation}
 and thus is a function of a single variable only.
 
 In the cases $\beta$ even (\ref{C1a}) exhibits a (Laurent) polynomial structure (in $e^{i \theta}$),
 \begin{equation}\label{C1c}
 \rho_{(2),N}^{\rm C}(\theta, 0) = \sum_{p=-\beta(N-1)/2}^{\beta (N-1)/2} \alpha_p(N,\beta) e^{i p \theta},
  \end{equation}
  where $\alpha_{-p}(N,\beta)  = \overline{ \alpha_p(N,\beta)}$.
  In an earlier work \cite{Fo93}, a differential-difference system was used to specify a recursive
  computational scheme for $\{  \alpha_p(N,\beta) \}$. This differential-difference system is part of the broader
  theory of the Selberg integral (see e.g.~\cite[Ch.~4]{Fo10}). Development of this theory by Kaneko
  \cite{Ka93} was applied in \cite{Fo94j} to deduce that for $\beta$ even the $(N-2)$-dimensional integral
  (\ref{C1}) permits a transformation to a $\beta$-dimensional integral,
   \begin{multline}\label{C1d}
   \rho_{(2),N}^{\rm C}(\theta, 0) = C_{N,\beta} |1 - e^{i \theta} |^\beta e^{-i \beta (N-2) \theta / 2} \\
    \times
   \int_{[0,1]^\beta} dx_1 \cdots dx_\beta \,
   \prod_{j=1}^\beta \Big ( x_j (1 - x_j) \Big )^{2/\beta - 1} \Big ( 1 - (1 - e^{i \theta}) x_j \Big )^{N-2}
   \prod_{1 \le j < k \le \beta} | x_k - x_j |^{4/\beta}.
   \end{multline} 
 The constant $ C_{N,\beta}$ is known explicitly, but is not needed in our subsequent working.
 
 Like the $(N-2)$-dimensional integral
  (\ref{C1}), the integral in (\ref{C1d}) also relates to a particular differential-difference system from the theory of the
  Selberg integral. We know from \cite{FR12} that the later is equivalent to a matrix differential equation of
  dimension $\beta + 1$; the differential-difference system applying directly to the integral  (\ref{C1}) is
  equivalent to a matrix differential equation of
  dimension $N - 1$. The advantage of the former is that, in theory at least, it allows the characterisation of the bulk scaled
  $N \to \infty$ limit of $  \rho_{(2),N}^{\rm C}(\theta, 0)$ in terms of a $(\beta + 1)$-dimensional scalar differential equation.
  Experience in the explicit determination of such differential equations for the eigenvalue density of the Gaussian
  $\beta$-ensemble \cite{RF19} shows that the practical implementation of the reduction from the matrix to the scalar
  differential equation requires computer algebra and involves cumbersome expressions for the coefficients,
  which however simplify in the scaling limit. In \cite{RF19}, this restricted the largest value of $\beta$ to be 
  considered to the value 6, corresponding to a 7th order differential equation. Here, by not seeking the form of the
  differential equation for finite $N$, but rather simplifying by taking the bulk scaled $N \to \infty$ limit using computer
  algebra, manageable expressions up to the value $\beta = 10$ can be obtained. However, for purposes of our application
  in deriving Proposition \ref{p5}, we don't go beyond $\beta = 8$.

  The differential-difference system from \cite{Fo93} relevant to (\ref{C1d}) applies to the family of
  multiple integrals
  \begin{multline}\label{5.1}
J_{p,n}^{(\alpha)}(z) = {1 \over C_p^n}\int_0^1 dx_1 \cdots  \int_0^1 dx_n \,
\prod_{l=1}^n x_l^{\lambda_1} (1 - x_l)^{\lambda_2} (z - x_l)^\alpha \\
\times \prod_{1 \le j < k \le n} | x_k - x_j|^{2 \kappa}
e_p(z - x_1,\dots, z - x_n),
\end{multline}
 where $e_p(y_1,\dots,y_n) $ denote the elementary symmetric polynomials in $\{ y_j \}_{j=1}^m$, and
$C_p^n$ denotes the binomial coefficient.
 This family of multiple integrals satisfies the
differential-difference system \cite{Fo93}, \cite[\S 4.6.4]{Fo10} (note that
these references have $\alpha$ replaced by $\alpha - 1$ relative to our (\ref{5.1})),
later observed to be equivalent to a certain Fuchsian matrix differential equation \cite{FR12},
\begin{equation}\label{5.1a}
(n - p) E_p J_{p+1}(z) 
= (A_p z + B_p) J_p(z) - z(z-1) {d \over dz} J_p(z) + D_p z ( z - 1) J_{p-1}(z),
\end{equation}
valid for $p=0,1,\dots, n $, where we have abbreviated $J_{p,n}^{(\alpha)}(z) =: J_p(z)$, and
\begin{align*}
A_p & = (n-p) \Big ( \lambda_1 + \lambda_2 + 2 \kappa  (n - p - 1) + 2(\alpha + 1) \Big ) \\
B_p & = (p-n)  \Big ( \lambda_1 + \alpha + 1 + \kappa (n - p - 1) \Big ) \\
D_p & = p \Big ( \kappa (n-p) + \alpha + 1 \Big ) \\
E_p & = \lambda_1 + \lambda_2 + 1 + \kappa (2n - p - 2) + (\alpha + 1).
\end{align*}
Recent applications of (\ref{5.1a}) and various scaling forms (for example $z \mapsto z/\lambda_2$ and
$\lambda_2 \to \infty$ which induces the Jacobi to Laguerre weight limit $x^{\lambda_1} (1 - x/\lambda_2)^{\lambda_2} \to
x^{\lambda_1} e^{- x}$) for purposes of exact finite $N$ computations can be found in \cite{FT18,Ku19,FT19,FK19}.

To be directly applicable to (\ref{C1d}) we consider the modifications of (\ref{5.1}) specified by
  \begin{multline}\label{5.1y}
K_{p,n}^{(\alpha)}(z) = {1 \over C_p^n} z^{- ( \alpha n +  \beta) /2} (z - 1)^\beta \int_0^1 dx_1 \cdots  \int_0^1 dx_n \,
\prod_{l=1}^n x_l^{\lambda_1} (1 - x_l)^{\lambda_2} (1 - (1 - z) x_l)^\alpha \\
\times \prod_{1 \le j < k \le n} | x_k - x_j|^{2 \kappa}
e_p(1 - (1 - z) x_1,\dots, 1 - (1 - z) x_n).
\end{multline}

  \begin{corollary}
 The family of multiple integrals (\ref{5.1y}) satisfy the differential-difference equation
 \begin{equation}\label{5.1f}
 (n - p) E_p K_{p+1}(z) 
= ( \tilde{A}_p  - z \tilde{B}_p  ) K_p(z) 
- z (1 - z) {d \over dz} K_p(z) + z D_p K_{p-1} (z),
\end{equation}
valid for $p=0,1,\dots, n $, where
\begin{equation}
\tilde{A}_p =  A_p + B_p - {\alpha n + \beta \over 2} , \qquad
\tilde{B}_p = B_p + p + {\alpha n + \beta \over 2}.
\end{equation}
 \end{corollary}
 
 \begin{proof}
 We first introduce the transformation of the multiple integrals (\ref{5.1})
 \begin{align}\label{5.1g}
 \tilde{J}_{p,n}^{(\alpha)}(z) & = z^{\alpha n + p} {J}_{p,n}^{(\alpha)} \Big ( {1 \over z} \Big ) \Big |_{z \mapsto 1 - z} \nonumber \\
& =  {1 \over C_p^n}\int_0^1 dx_1 \cdots  \int_0^1 dx_n \,
\prod_{l=1}^n x_l^{\lambda_1} (1 - x_l)^{\lambda_2} (1 - (1 - z) x_l )^\alpha  \nonumber \\
& \qquad \times \prod_{1 \le j < k \le n} | x_k - x_j|^{2 \kappa}
e_p(1 - (1-z) x_1,\dots, 1 - (1-z)x_n).
\end{align}
Substitution into (\ref{5.1a}) shows
\begin{multline}\label{5.1b}
(n-p) E_p  \tilde{J}_{p+1}(z) = \Big ( (A_p + B_p ( 1 - z)) - z ( \alpha n + p) \Big )  \tilde{J}_{p}(z)  \\
- z (1 - z) {d \over d z} \tilde{J}_{p}(z)
 + z D_p  \tilde{J}_{p-1}(z).
 \end{multline}
 Comparing the multiple integral in (\ref{5.1g}) to (\ref{5.1y}) we see that
 $$
K_{p,n}^{(\alpha)}(z) =   z^{- ( \alpha n +  \beta) /2} (z - 1)^\beta   \tilde{J}_{p,n}^{(\alpha)}(z).
$$
Substituting this in (\ref{5.1b}) gives the further transformed recurrence (\ref{5.1f}).
  \end{proof}
 
Introduce the $(n+1) \times (n+1)$ tridiagonal matrix
$$
X = {\rm diag} \, [\tilde{A}_p  - z \tilde{B}_p ]_{p=0}^n -  {\rm diag}^+ \, [(n-p) E_p ]_{p=0}^{n-1}  +   {\rm diag}^- \, [z D_p]_{p=1}^n,
$$
where $ {\rm diag}^+$ refers to the first diagonal above the main diagonal, and 
${\rm diag}^-$ refers to the first diagonal below the main diagonal.
Introduce too the vector of multiple integrals $\mathbf v = [ K_{p,n}^{(\alpha)}(z) ]_{p=0}^n$.
We see immediately that the recurrences (\ref{5.1f}) are equivalent to the matrix differential
equation
\begin{equation}\label{5.1h}
z (1 - z) {d \over d z} \mathbf v = X \, \mathbf  v.
\end{equation}
Moreover, the tridiagonal structure allows for a straightforward reduction from the first order
matrix differential equation for the vector $\mathbf v$ to a scalar differential equation of order
$n+1$ for its first component $K_{0,n}^{(\alpha)}(z) =: K_0(z)$.

To see this, note that the first row in (\ref{5.1h}), or equivalently the recurrence
(\ref{5.1f}) with $p=0$, tells us that
\begin{equation}\label{K1}
n E_0 K_1(z) = (\tilde{A}_0 - z \tilde{B}_0 ) K_0(z) - \delta_z K_0(z), \quad \delta_z := z (1 - z) {d \over d z},
\end{equation}
hence expressing $K_1(z)$ in terms of $K_0(z)$ and its derivative. The recurrence
(\ref{5.1f} with $p=1$ tells us that
\begin{equation}\label{K2}
(n-1) E_1 K_2(z) = (\tilde{A}_1 - z \tilde{B}_1 ) K_1(z) - \delta_z K_1(z) + z D_1 K_0(z).
\end{equation}
Substituting (\ref{K1}) in (\ref{K2}) gives an expression for $K_2(z)$ in terms of $K_0(z)$ and
its first two derivatives. Continuing this procedure, an expression for $K_p(z)$ in terms of
$\{ {d^k \over d z^k} K_0(z) \}_{k=0}^p$ can be obtained for each $p=1,2,\dots,n$.
Substituting the cases $p=n$ and $p=n-1$ into the recurrence (\ref{5.1f}) with $p=n$,
$$
0 = (\tilde{A}_n - z \tilde{B}_n) K_n(z) - z (1 - z) {d \over dz} K_n(z) + z D_n K_{n-1}(z),
$$
gives the sought scalar differential equation for $K_0(z)$, which is seen to be of order $n+1$.

We see that $K_{p,n}^{(\alpha)}(z)$ as specified by (\ref{5.1y}) is proportional to $\rho_{(2),N}(\theta,0)$
as specified by (\ref{C1d}) if we set
$$
p=0, \quad \alpha = N-2, \quad n = \beta, \quad \lambda_1 = \lambda_2 = {2 \over \beta} - 1, \quad \kappa =  2/\beta, \quad z = e^{i \theta}.
$$
Since $\beta$ is required to be even, the lowest order differential equation we can obtain
for $\rho_{(2),N}(\theta,0)$ occurs when $\beta = 2$, and is order 3. However, already for
this simplest case the dependence on $N$ in the coefficients makes for a cumbersome
expression. On the other hand, our primary interest in $\rho_{(2),N}^{\rm (C)}(\theta, 0)$ is
not its dependence on $N$, but rather its bulk scaling limit
\begin{equation}\label{K3}
\rho_{(2),\infty}^{\rm (C)}(X, 0) = \lim_{N \to \infty} \Big ( {2 \pi \over N} \Big )^2 
\rho_{(2),N}^{\rm (C)}\Big ( 2 \pi  X / N, 0 \Big ).
\end{equation}
The existence of this limit, and moreover its explicit form as a $\beta$-dimensional integral,
follows immediately from (\ref{C1d})  \cite{Fo94j}, \cite[Eq.~(13.35)]{Fo10}; in fact starting with
 (\ref{C1d}) it is possible to show that the rate
of convergence to the limit is O$(1/N^2)$ \cite{FLT20}.
The significance of the existence of the limit, and
moreover with corrections that are inverse powers of $N$, is that the differential operator
characterising $\rho_{(2),N}(\theta,0)$ can similarly be expanded, with the leading order
giving the differential operator for $\rho_{(2),\infty}^{\rm (C)}(X, 0) $. 

\subsection{The cases $\beta = 2$ and 4}
Through the use of
computer algebra, we first carry out the above procedure 
for $\beta = 2$ and 4, when the functional form is already known in terms of elementary functions.

\begin{proposition}\label{p2}
For $\beta = 2$, $\rho_{(2),\infty}^{\rm (C)}(X, 0) $ satisfies the 3rd order differential equation
\begin{equation}\label{K4a}
X^3   f^{(3)}(X)  + 4 X^2   f^{(2)}(X)  + \Big ( (2 \pi X )^2 - 2 \Big ) X f'(X)
- 4  f(X) = 0,
\end{equation}
while for $\beta = 4$, $\rho_{(2),\infty}^{\rm (C)}(X, 0) $ satisfies the 5th order  differential equation
\begin{multline}\label{K4b}
X^5 f^{(5)}(X) + 10 X^4 f^{(4)}(X)  + (20 \pi^2 X^5 + 12 X^3 )  f^{(3)}(X) + (64 \pi^2 X^4 - 40 X^2 )  f^{(2)}(X) \\
 + (64 \pi^4 X^5 - 48 \pi^2 X^3 - 16 X) f'(X) 
+ (-32 \pi^2 X^2 + 16) f(X) = 0.
\end{multline}
\end{proposition}

As a check, we recall (see e.g.~\cite[Eq.~(7.2)]{Fo10}) the explicit functional form
\begin{equation}\label{K4c}
\rho_{(2),\infty}^{\rm (C)}(X, 0) \Big |_{\beta = 2} = 1 - \Big ( {\sin \pi X \over \pi X} \Big )^2 =
1 - {1 \over 2 (\pi X)^2} +  {\cos 2 \pi X \over \pi^2 X^2}.
\end{equation}
It's straightforward to verify that the three linearly independent solutions of (\ref{K4a}) can
be chosen as
\begin{equation}\label{K5}
1 - \Big ( {1 \over 2 \pi X} \Big )^2, \qquad {\cos 2 \pi X \over  X^2}, \qquad {\sin 2 \pi X \over  X^2}.
\end{equation}
To specify the particular solution corresponding to $\rho_{(2),\infty}^{\rm (C)}(X, 0) $ we need to
know the scalars in the linear combination. The bulk scaling limit corresponds to a unit density,
so we must have $\rho_{(2),\infty}^{\rm (C)}(X, 0) \to 1$ as $X \to \infty$, telling us that the scalar
multiplying the first solution in (\ref{K5}) is unity. The coefficient of the third solutions must be
zero, since it is odd in $X$. And requiring that $\rho_{(2),\infty}^{\rm (C)}(X, 0)$ be bounded as
$X \to 0$ (in fact it must vanish as $X^2$) tells us that the coefficient of the second term is
$1/\pi^2$. The indicial equation for (\ref{K4a}) has roots $\{-2,-1,2\}$, and so
 (\ref{K4c}) is the Frobenius series solution corresponding to exponent 2.
 
 \begin{remark} Consider the CUE conditioned so that there is an eigenvalue fixed
 at $\theta = 0$. Let $\tilde{\rho}_{(k)}(X_1,\dots,X_k)$ denote the corresponding
 $k$-point correlation function in the bulk scaling limit with $\tilde{\rho}_{(1)}(X)
 \sim 1/2 \pi$ as $X \to \infty$ (the normalisation of $1/2 \pi$ rather than unity is for convenience).
 We know from \cite{FW04}, \cite[Prop.~8.3.4, with $\mu=0, \omega = 1$]{Fo10} that the
 series
 $$
 1 + \sum_{k=1}^\infty {(-\xi)^k \over k!} \int_0^x d X_1 \cdots  \int_0^x d X_k \, \tilde{\rho}_{(k)}(X_1,\dots,X_k),
 $$
 which for $\xi = 1$ is equal to the probability that the interval $(0,x)$ is eigenvalue free
 (see \cite[\S 9.1]{Fo10}), has the $\tau$-function form
 $$
 \exp \int_0^x h(-is) \, {ds \over s},
 $$
 where $h(t)$ satisfies the $\sigma$PV equation (see \cite[Eq.~(8.15)]{Fo10}) with parameters
 $\nu_0 = \nu_1 = 0$, $\nu_2 = - \nu_3 = 1$,
 \begin{equation}\label{jh}
 (t \sigma'')^2 - \Big (\sigma - t \sigma' + 2 ( \sigma')^2 \Big )^2 +
 4 (\sigma')^2 \Big ( (\sigma')^2 - 1 \Big ) = 0.
 \end{equation}
 Notice that (\ref{jh}) is independent of $\xi$. This appears only in the boundary condition, which
 reads $h(-it) \sim - \xi t \rho_{(1)}(t)$ as $t \to 0^+$. As in the relationship between
 the third order linear differential equations derived for the GUE and LUE and particular
 $\sigma$ Painlev\'e equations in \cite[\S 3.3]{FT18}, and similarly for the JUE \cite[Remark 2.3]{RF19},
 it follows by changing variables $t \mapsto - it$, substituting the boundary condition in (\ref{jh}), 
 and equating leading order
 in $\xi$ that $t \rho_{(1)}(t)$ satisfies a nonlinear equation
 \begin{equation}\label{jh1}
 (t \sigma'')^2 + (\sigma - t \sigma')^2 - 4 (\sigma')^2 = 0.
  \end{equation}
  Differentiating this with respect to $t$ gives a third order linear differential equation,
  which with the substitutions $\sigma(t) = t \rho_{(1)}(t)$ and $t \mapsto 2 \pi t$ is
  precisely (\ref{K4a}).
  \end{remark}

In relation to the differential equation (\ref{K4b}), we know that $\rho_{(2),\infty}^{\rm (C)}(X, 0) |_{\beta = 4}$
can be expressed in terms of the sinc function (see e.g.~\cite[Eq.~(7.92)]{Fo10}),
\begin{equation}\label{K4c+}
\rho_{(2),\infty}^{\rm (C)}(X, 0) \Big |_{\beta = 4} = 1 - \Big ( {\sin 2 \pi X \over 2 \pi X} \Big )^2 +
{1 \over 2 \pi} \Big (  {\partial \over \partial X} {\sin 2 \pi X \over 2 \pi X} \Big )
\int_0^{2 \pi X} {\sin t \over t } \, dt.
\end{equation}
Two linearly independent elementary solutions of (\ref{K4b}) are
\begin{equation}\label{K4d}
{\sin 2 \pi X \over 4 X} + {\cos 2 \pi X \over 8 \pi X^2}, \qquad
{\cos 2 \pi X \over 4 X} - {\sin 2 \pi X \over 8 \pi X^2}.
\end{equation}
These can be supplemented by  solutions which  for $X \to \infty$ have the structures
\begin{equation}\label{K4e}
1 - {1 \over  (2\pi X)^2} \bigg ( 1 + {\rm O} \Big ( {1 \over X^2} \Big ) \bigg ),
\end{equation}
and
\begin{equation}\label{K4e+}
{e^{\pm 4 \pi i X} \over X^4} \bigg ( 1 \pm i {c_1 \over X} +  {\rm O} \Big ( {1 \over X^2} \Big ) \bigg ).
\end{equation}
Of the five linearly independent solutions in (\ref{K4d}), (\ref{K4e}) and (\ref{K4e+}),
the first in
(\ref{K4d}) and the imaginary part of (\ref{K4e+}) do not
contribute to the $X \to \infty$ asymptotic expansion of (\ref{K4c+}), which
reads
\begin{multline}\label{K4f}
\rho_{(2),\infty}^{\rm (C)}(X, 0) \Big |_{\beta = 4}  \mathop{\sim}\limits_{X \to \infty} 
1 - {1 \over  (2\pi X)^2} \bigg ( 1 + {\rm O} \Big ( {1 \over X^2} \Big ) \bigg )
+ {\cos 2 \pi X \over 4 X} - {\sin 2 \pi X \over 8 \pi X^2}  \\
 + {  \cos 4 \pi   X  \over 32 \pi^4 X^4 }   \bigg ( 1 + {\rm O} \Big ( {1 \over X^2} \Big ) \bigg ) 
+ { \sin 4 \pi   X \over 16 \pi^5 X^5}  \bigg ( 1 + {\rm O} \Big ( {1 \over X^2} \Big ) \bigg ).
\end{multline}

Also significant from the viewpoint of (\ref{K4b}) is the behaviour as $X \to 0$. We can check from (\ref{K4c+}) that 
\begin{equation}\label{K4g}
\rho_{(2),\infty}^{\rm (C)}(X, 0) \Big |_{\beta = 4}  = {(2 \pi X)^4 \over 135} - {2 (2 \pi X)^6 \over 4725} + {\rm O} (X^8).
\end{equation}
With the roots of the indicial equation of (\ref{K4b}) being $\{ -2,-1,1,4\}$, with $-2$ repeated, it follows
that (\ref{K4c+}) is the Frobenius series solution corresponding to the exponent 4.

\begin{remark}\label{R3}
We observe that
\begin{equation}\label{tm}
{\partial \over \partial X} \bigg ( {\sin 2 \pi X \over 2 \pi X} \bigg ) =
{\cos 2 \pi X \over X} - {\sin 2 \pi X \over 2 \pi X^2}
\end{equation}
and is thus proportional to the second of the elementary solutions of
(\ref{K4b}) as listed in (\ref{K4d}). We can therefore add a scalar multiple
of (\ref{tm}) to (\ref{K4c+}) --- the scalar we will take to equal $-{1 \over 2}$ --- to obtain
the further solution of (\ref{K4b})
\begin{equation}\label{tR}
R(X)  = 1 - \Big ( {\sin 2 \pi X \over 2 \pi X} \Big )^2 +
{1 \over 2 \pi} \Big (  {\partial \over \partial X} {\sin 2 \pi X \over 2 \pi X} \Big )
\bigg ( - {\pi \over 2} + 
\int_0^{2 \pi X} {\sin t \over t } \, dt \bigg ).
\end{equation}
We recognise this (see e.g.~\cite[Eq.~(7.133)]{Fo10}) as $\rho_{(2),\infty}^{\rm (C)}(2X, 0) \Big |_{\beta = 1} $,
restricted to $X > 0$.
\end{remark}

A closely related  inter-relationship between $\rho_{(2),\infty}^{\rm (C)}$ for $\beta = 1$ and $\beta = 4$ is known
from \cite{FJM00}.  Thus for general $\beta > 0$ define the structure function
\begin{equation}\label{tR1}
S(k;\beta) = \int_{-\infty}^\infty \Big ( \rho_{(2),\infty}^{\rm (C)}(x,0) - 1 + \delta(x)\Big ) e^{i k x} \, dx,
\end{equation}
which according to (\ref{2a})  is equivalent to the definition (\ref{5}), and define the function $f(k;\beta)$ by (\ref{tR2}).
Then we know from \cite{FJM00} that
\begin{equation}\label{tR3}
f(k;\beta) = f \Big ( - {2k \over \beta} ; {4 \over \beta} \Big ).
\end{equation}

However, the restricted range in (\ref{tR2}) is a necessary consequence of the Fourier expansion
in (\ref{C1c}) terminating, whereas the corresponding expansion in the case $\beta = 1$
does not. For finite $N$, we know from \cite[Prop.~4.7]{WF14} that for $p \ne 0$ the
coefficients in (\ref{C1c}) modified by the addition of $N$ and so in terms of
$\tilde{\alpha}_p(N,\beta) := {\alpha}_p(N,\beta) + N$, satisfies the functional
equation
\begin{equation}\label{tR4}
 \tilde{\alpha}_p(N,\beta)  = (4/\beta^2)  \tilde{\alpha}_p(- \beta N/2,4/\beta),
\end{equation}
whenever both sides are nonzero.
This was later observed \cite[Appendix B and Eq.~(4.15)\footnote{Duy Khanh
Trinh (personal correspondence) has pointed out that the normalising factor $(-\kappa)^k$ on
the RHS should be deleted and replaced by a factor of $-(1/\kappa)$.}]{FRW17} to
follow from the functional equation for the moments of the Jacobi $\beta$-ensemble
\cite{DP12}.

If (\ref{tR4}) was valid for all $p \ne 0$ a characterising equation for
$\rho_{(2),N}^{\rm (C)}(X, 0) \Big |_{\beta = 1} $, would be obtained from
that for $\rho_{(2),N}^{\rm (C)}(X, 0) \Big |_{\beta = 4} $ by simply replacing
$N$ by $-N/2$; see \cite{WF14}, \cite{RF19} for application of this strategy in 
relation to linear differential equations for the classical Gaussian, Laguerre and
Jacobi $\beta$-ensembles. Taking the limit $N \to \infty$ would then give the finding that
(\ref{tR}) satisfies (\ref{K4b}).

\begin{remark}
In the sense of generalised functions, it follows from (\ref{tR1}) that
$$
S(k;\beta) = \int_{-\infty}^\infty \rho_{(2)}^{(C)}(x,0) e^{i k x} \, dx- 2 \pi \delta(k) + 1;
$$
here $ \rho_{(2)}^{(C)}(x,0)$ is to be regularised at infinity by multiplication by a rapidly
decaying function, for example $e^{- \mu x^2}$, with the limit $\mu \to 0^+$ taken after
the integration. A differential equation for $g(k) := \int_{-\infty}^\infty  \rho_{(2),N}^{\rm (C)}(x, 0) 
e^{i k x} \, dx$ so interpreted can, in the case $\beta = 4$, be deduced from 
(\ref{K4b}) upon multiplying through by $e^{i k X}$, integrating over $X$ and simplifying using
integration by parts. We find
\begin{multline}\label{Nv}
(-64k \pi^4 + 20 k^3 \pi^2 - k^5) g^{(5)}(k) +
(-320 \pi^4 + 236 k^2 \pi^2 - 15k^4) g^{(4)}(k) \\
+ (-52 k^3 + 640 k \pi^2)  g^{(3)}(k) + (-28 k^2 + 320 \pi^2)  g^{(2)}(k) =0;
\end{multline}
note the absence of terms involving $g'(k)$ or $g(k)$. However for $\beta = 4$ we know the explicit evaluation
(see e.g.~\cite[Eq.~(7.95)]{Fo10})
$$
S(k;\beta) \Big |_{\beta = 4} =
\left \{
\begin{array}{ll} {k \over 4 \pi} - {k \over 8 \pi} \log | 1 - {k
\over 2 \pi}|, &
 0 < k \le 4 \pi, \\[.2cm]
1, &
k \ge 4 \pi. \end{array} \right.
$$
This indeed can be checked to satisfy (\ref{Nv}). We remark too that since (\ref{Nv}) is unchanged by replacing $k$ by
$-k$, we have that $S(-k;\beta)  |_{\beta = 4} $ also satisfies (\ref{Nv}).
\end{remark}

\subsection{The case $\beta = 6$}
The use of computer algebra to carry out the procedure of subsection \S \ref{s2.1} gives
the analogue of Proposition \ref{p2} for $\beta = 6$.

\begin{proposition}\label{p4}
For $\beta = 6$, $\rho_{(2),\infty}^{\rm (C)}(X, 0) $ satisfies the 7th order linear differential equation
\begin{multline}\label{J4a}
9X^7 f^{(7)}(X)+168X^6f^{(6)}(X)+(700+504 \pi^2 X^2)X^5 f^{(5)}(X) +(-692+4704 \pi^2 X^2)X^4 f^{(4)}(x)\\
+(-3924+4080 \pi^2 X^2 +7056 \pi^4 X^4)X^3f^{(3)}(X) 
+(2688-18144 \pi^2 X^2 +20736 \pi^4 X^4) X^2 f^{(2)}(x) \\+(2808+144 \pi^2 X^2 -17280 \pi^4 X^4 +20736 \pi^6 X^6)X f'(X)  \\
+(-1008+6336 \pi^2 X^2 -6912 \pi^4 X^4)f(X) = 0.
\end{multline}
\end{proposition}

\begin{remark}\label{R11a}
Our (computer algebra assisted) method of deduced (\ref{J4a}) can be extended to higher even
$\beta$ values. For example, for $\beta = 8$ we find that $\rho_{(2),\infty}^{\rm (C)}(X, 0) $ satisfies the 9th order 
homogeneous linear differential equation
\begin{multline}\label{r11b}
16 X^9 f^{(9)}(X) + 480 X^8 f^{(8)}(X)  + (4056   + 1920 \pi^2 X^2) X^7  f^{(7)}(X) +  \cdots  \\
+ (20000  - 259584 \pi^2 X^2 + 688128 \pi^4 X^4  - 
 589824 \pi^6 X^6 ) f(X)= 0,
\end{multline}
where the writing of all terms involving lower order derivatives has been suppressed, due to the
increasing complexity with respect to the number of monomials in their coefficients and the
size of the corresponding scalars.
\end{remark}

We can check that the seven linearly independent solutions of (\ref{J4a}) can be distinguished according to the
$X \to \infty$ behaviours
\begin{align}\label{J4b}
&1 - {1 \over 6 (\pi X)^2 }+ {\rm O} \Big ( {1 \over X^4} \Big ), \\
& {e^{\pm 2 \pi i p X} \over X^{2 p^2/3}} \bigg ( 1 \pm i {4 p \over 3} \Big ( 1 - {p^2 \over 3} \Big ) {1 \over X} +
{\rm O} \Big ( {1 \over X^2} \Big ) \bigg ), \quad (p=1,2,3). \label{J4c}
\end{align}
Using the integral formula for  $\rho_{(2),\infty}^{\rm (C)}(X, 0) $ with $\beta = 6$ which follows from
(\ref{C1d}), the asymptotic expansion of the two-point correlation function is known from 
\cite{Fo93aa} to be given by the first of these forms, plus a linear combination of the real part of the complex
solutions.

One advantage of knowledge of the differential equation over the integral formula is that the former is
readily suited to extending the asymptotic expansion. Of particular interest are higher order 
non-oscillatory terms, and thus the extension of (\ref{J4b}).

\begin{corollary}
The solution of (\ref{J4a}) corresponding to  (\ref{J4b}) has the $X \to \infty$
asymptotic expansion
\begin{multline}\label{J4c}
1 - {1 \over 6 (\pi X)^2 } + {1 \over 9 (\pi X)^4} -  {55 \over 162 (\pi X)^6} + {5215 \over 1944 (\pi X)^8} - {17105 \over 432 (\pi X)^{10}} \\
+ {681505 \over 729 (\pi X)^{12}} - {140887175 \over 4374 (\pi X)^{14}}  + {\rm O} \Big ( {1 \over X^{16}} \Big ).
\end{multline}
\end{corollary}

Knowledge of the non-oscillatory $X \to \infty$ expansion of $\rho_{(2),\infty}^{\rm (C)}(X, 0) $ gives information on the form of the
$k \to 0$ expansion of the structure function (\ref{tR1}). Thus, as made explicit previously in \cite[\S 4]{FJM00}, the non-oscillatory
$X \to \infty$ expansion
\begin{equation}\label{f1}
\rho_{(2),\infty}^{\rm (C)}(X, 0)  \mathop{\sim}\limits_{X \to \infty} 
1 + \sum_{n=1}^\infty {c_n \over X^{2n}}
\end{equation}
is equivalent to the singular small-$k$ expansion of the structure function
\begin{equation}\label{f2}
S(k;\beta)   =  \pi \sum_{n=1}^\infty
{(-1)^n c_n \over (2n-1)! } |k|^{2n-1}
\end{equation}
(in distinction to the expansion (\ref{f1}), this series has a nonzero radius of convergence).
Another viewpoint on (\ref{f2}) is that, after multiplying by $\pi \beta / |k|$, it gives
the even part of the power series of the analytic function $f(k;\beta)$ (\ref{tR2}).
Previous work has determined the power series expansion of $f(k;\beta)$, up to and
including the coefficient of $k^9$ for general $\beta > 0$; each coefficient of $k^j$ is a
polynomial of degree $j$ in $2/\beta$. In particular, with $x:= \beta / 2$, the coefficient
of $k^8$ is \cite[\S 8]{FJM00}, \cite[Remark 4.1]{WF15}
\begin{equation}\label{f3}
\Big ( {1 \over 2 \pi x} \Big )^8 (x - 1)^2 \Big (
x^6 - {263 \over 84} x^5 + {1697 \over 315} x^4 -
{6337 \over 1008} x^3 + {1697 \over 315} x^2- {263 \over 84} x + 1 \Big ).
\end{equation}
Multiplying by $1/(2\pi x)$, substituting $\beta = 6$ and recalling (\ref{f2}) gives the value of $c_{10}$ as read off from (\ref{J4c}).

\begin{remark}\label{R6a}
From the discussion of Remark \ref{R3} we conjecture that $\rho_{(2),\infty}^{\rm (C)}(3X, 0) \Big |_{\beta = 2/3}$
satisfies the differential equation (\ref{J4a}). Recalling the definition (\ref{tR2}) we see that the replacement
$X \mapsto 3 X$ in (\ref{f1}), and thus $c_n \mapsto c_n/3^{2n}$, is when substituted in (\ref{f2}) consistent with
the functional equation (\ref{tR3}).
\end{remark}

Also compelling in favour of the conjecture in Remark \ref{R6a}
are the forms of the small $X$ solutions permitted by (\ref{J4a}).
Thus direct substitution shows there are solutions
\begin{equation}\label{S1y}
X^6 \Big ( 1 - {18 \pi^2 \over 55} X^2 + {144 \pi^4 \over 2695} X^4 - {3456 \pi^6 \over 595595} X^6 + \cdots \Big )
\end{equation}
and (for $X>0$)
\begin{equation}\label{S2y}
X^{2/3} \Big ( 1 - {6 \pi^2 \over 5} X^2 + {72 \pi^4 \over 175} X^4 - {144 \pi^6 \over 1925} X^6 +  \cdots \Big ).
\end{equation}
The first of these is the small $X$ form of $\rho_{(2),\infty}^{\rm (C)}(X, 0) |_{\beta = 6}$. The second is consistent
with the leading small $X$ form of $\rho_{(2),\infty}^{\rm (C)}(3X, 0) |_{\beta = 2/3}$ (in contrast to the large 
$X$ form for $\beta$ rational \cite{Ha95}, \cite[\S 13.7.3]{Fo10}, the structure of the small $X$ expansion to
all leading independent powers appears not to be documented in the literature). In relation to this last point,
we can check that the roots of the indicial equation corresponding to (\ref{J4a}) are $\{-7/3,-2,-1,2/3,3,6\}$
with $-2$ of multiplicity 2.  Thus, assuming the validity of the conjecture, the small $X$ expansion
of $\rho_{(2),\infty}^{\rm (C)}(3X, 0) |_{\beta = 2/3}$ can consist of a linear combination of (\ref{S2y})
and  $X^3$ times a power series in $X$.

\subsection{Coefficient of $k^{11}$ in $S(k;\beta)$}
As an application of the expansion (\ref{J4c}), we will show how knowledge of the coefficient of $1/X^{12}$,
when combined with already established results from \cite{FJM00},
allows us to deduce the coefficient of $k^{10}$ in the expansion of $f(k;\beta)$.

First, we know from  \cite{FJM00}, \cite{WF15} that the sought coefficient has a polynomial structure
\begin{equation}\label{f3}
[ k^{10}] {\beta \pi \over |k|} S(k;\beta) =
\Big ( {1 \over 2 \pi x} \Big )^{10} (x - 1)^2 \sum_{l=0}^8 b_{10,l} x^l \quad {\rm with} \quad b_{10,l} = b_{10,8-l},
\end{equation}
where $[ k^{10}] h(k)$ denotes the coefficient of $k^{10}$ in $h(k)$, and $x = \beta/2$ as in (\ref{f3}).
Results in \cite{FJM00} also tell us that
\begin{align}\label{f4}
1 & = b_{10,0} \nonumber \\
 -\sum_{q=1}^{11} {1 \over q} (1 - 2^{q-11}) & = {1 \over 2} (b_{10,1} - 2),
 \end{align}
  which are obtained by performing an exact expansion about $\beta = 0$;
  \begin{equation}\label{f5}
 {1 \over 5}  = \Big ( (1 + 2^{-8}) + (2^{-1} + 2^{-7})b_{10,1} + (2^{-2} + 2^{-6})b_{10,2} + 
 (2^{-3} + 2^{-5})b_{10,3}  + 2^{-4} b_{10,4} \Big ),
  \end{equation}
  which follows from knowledge of the Fourier transform of (\ref{K4c+});
   \begin{equation}\label{f6}
   [|k|^{11}] {\partial S(k;\beta) \over \partial \beta} \Big |_{\beta = 4} =
   {1949 \over 275251200 \pi^{11}},
 \end{equation}   
which follows from performing an exact expansion about $\beta = 4$  \cite[Eq.~(7.22)]{FJM00} and gives
  \begin{equation}\label{f7}
 {1949 \over 275251200 } = 4^{-11}  \sum_{l=0}^8 b_{10,l} 2^{l }
 \Big ( 1 + {1 \over 4} (l - 11) \Big ).
  \end{equation}

The new information comes from  (\ref{f2}) and (\ref{J4c}) to deduce the value of the left hand side of
(\ref{f3}) for $\beta = 6$, telling us that
 \begin{equation}\label{f8}
 {681505 \over 729} = {11! \over 6^{11}} 2^2 \sum_{l=0}^8 b_{10,l} 3^l.
   \end{equation} 
   
   In (\ref{f3}) the coefficient $b_{10,0}$ and $b_{10,1}$ are known immediately from (\ref{f4}), leaving
   three unknowns $b_{10,j}$, $j=2,3,4$. The equations (\ref{f6}), (\ref{f7}) and (\ref{f8}) give three
   independent linear equations for these unknowns. Solving, we have fully determined the polynomial
   on the right hand side of (\ref{f3}), and Proposition \ref{p5} as stated in Section \ref{S1.1} follows.

 \begin{remark}
1.  As is the case for $[ k^{j}] {\beta \pi \over |k|} S(k;\beta) $, $j=1,\dots,9$ the polynomial
 part of (\ref{f9}) has all zeros on the unit circle in the complex $x$-plane. Moreover, the
 zeros in the upper half plane of the sixth order polynomial in (\ref{f3}) interlace the zeros in
 the upper half plane of the eighth order polynomial in (\ref{f9}). For analogous observations
 in relation to the moments of the spectral density of the Gaussian $\beta$-ensemble,
 see \cite{WF14}. \\
  2. A computer algebra calculation obtained by
  substituting the non-oscillatory asymptotic expansion (\ref{f1}) in
  the differential equation (\ref{r11b})
  gives that for $\beta = 8$
  we have for the values of the coefficients
   \begin{equation}\label{12.3}
   \dots, c_{6} = {19405708245 \over 16777216 \pi^{12}}, \quad
 c_{7} =  -{11022926679765 \over 268435456 \pi^{14}}, \dots
  \end{equation} 
  From   (\ref{f2}) we know $c_{6}$ relates to $[ k^{10}] {\beta \pi \over |k|} S(k;\beta) $, which
 in turn is given by  Proposition \ref{p5} . Substituting
   $\beta = 8$ gives agreement with (\ref{12.3}).
 
  \end{remark}
 
 \section{The averaged structure function for the GUE}
 We now turn our attention to the quantity $\bar{S}(k)$ in the case of the GUE.
 Associated with this ensemble is the classical weight function $w(x) = e^{-x^2}$. 
 The corresponding monic orthogonal polynomials are
 given in terms of the  Hermite  polynomials by
 \begin{equation}\label{GL}
 p_n^{(G)}(x) = 2^{-n} H_n(x) .
 \end{equation}
 In terms of $\psi_n(x) = \sqrt{w(x)} p_n(x)$, these polynomials satisfy the matrix differential recurrence
  \begin{equation}\label{GLa}
   \begin{bmatrix} \psi_n'(x) \\ \psi_{n-1}'(x)  \end{bmatrix} =
   \begin{bmatrix}  x &  n \\
   2  &   -x \end{bmatrix}   \begin{bmatrix} \psi_n(x) \\ \psi_{n-1}(x)  \end{bmatrix}.
   \end{equation}
    In the context of gap probabilities of random matrix ensembles with a unitary symmetry,
   which includes the GUE, the significance of (\ref{GLa}) was highlighted in \cite{TW94c}.
   
  The GUE is an example of a determinantal point process, since its $k$-point correlation function
  $\rho_{(k)}(x_1,\dots,x_k)$ has the structure
   \begin{equation}\label{kp}
   \rho_{(k)}(x_1,\dots,x_k) =  \det \Big [ K_N(x_j,x_k) \Big ]_{j,l=1}^k.
   \end{equation}
   The function $K_N(x,y)$ is known as the correlation kernel, and is given in terms of the functions
   $\psi_n(x)$ specified below (\ref{GL}) according to
   \begin{equation}\label{kp1x}  
  K_N(x,y) = \sum_{n=0}^{N-1} { \psi_n(x) \psi_n(y)  \over (p_n,p_n)} = {1 \over (p_{N-1},p_{N-1})} {\psi_N(x) \psi_{N-1}(y) - 
  \psi_N(y) \psi_{N-1}(x)  \over x - y},
 \end{equation}  
 where $(f,g) = \int_{-\infty}^\infty w(x) f(x) g(x) \, dx$ is the inner product associated with $w$.
The second equality is the Christoffel-Darboux formula; see e.g.~\cite[Prop.~5.1.3]{Fo10}.
The explicit value of the required inner product in (\ref{kp1}) is (see e.g.~\cite[Eq.~(5.48)]{Fo10})
  \begin{equation}\label{kp1y} 
  (p_{N-1}, p_{N-1})^{(G)} = \sqrt{\pi} 2^{-N+1} (N-1)! .
  \end{equation}
  
  As is evident from (\ref{bhX}), the study of $\bar{S}(k)$ also involves the Laguerre unitary ensemble.
  Associated with this ensemble is the weight $w(x) = x^a e^{-x}$ ($x > 0$). Again setting $\psi_n(x) =
  \sqrt{w(x)} p_n(x)$, where now $p_n(x)$ refers to the monic orthogonal polynomial of degree $n$
  associated with the Laguerre weight, the correlation functions are again given by the determinant
  formula (\ref{kp}) with correlation kernel (\ref{kp1x}). The explicit form of $p_n(x)$, and the
  corresponding inner product required in (\ref{kp1x}), is
  \begin{equation}\label{kp1z} 
 p_n^{(L)}(x) = (-1)^n n! L_n^{(a)}(x),  \quad  (p_{N-1}, p_{N-1})^{(L)} =\Gamma(N) \Gamma(a+N).
  \end{equation} 

Combining (\ref{GLa}) with the second equality in (\ref{kp1x}) allows for the derivation of the differential identity
\cite{TW94c}, \cite[\S 5.4.2]{Fo10}
\begin{equation} \label{A1a} 
\Big ( {\partial \over \partial x} + {\partial \over \partial y} \Big )
K_N^{(G)}(x,y)  = - {1 \over  (p_{N-1},p_{N-1})^{(G)}} \Big (   \psi_N^{(G)}(x) \psi_{N-1}^{(G)}(y) +  \psi_{N-1}^{(G)}(x) \psi_{N-1}^{(G)}(y) \Big ) . 
\end{equation}
Taking the limit $x \to y = t$ is this formula gives
\begin{equation}
{d \over d t} K_N^{(G)}(t,t)  = - {2 \over  (p_{N-1},p_{N-1})^{(G)}}   \psi_N^{(G)}(t) \psi_{N-1}^{(G)}(t).  \label{A1} 
\end{equation}
A direct verification of (\ref{A1}) can also be given; see \cite[Lemma 2.3]{HT03}, \cite[Lemma 5.1]{Le04},
\cite[Eqns.~(2.48)]{GT05}.

\subsection{$\bar{S}(k)$ for the GUE}
Making use of (\ref{1})--(\ref{2a}), and (\ref{kp}) with $k=1$ and 2 in (\ref{8b}) shows
 \begin{equation}\label{kp1}
 \bar{S}^{(G)}(k) =  N - \int_{-\infty}^\infty d \lambda \, e^{-i k \lambda}   \int_{-\infty}^\infty d \lambda' \, e^{i k \lambda'} \Big ( K^{(G)}(\lambda, \lambda') \Big )^2.
 \end{equation} 
 Also worthy of independent attention is the average
  \begin{equation}\label{kp2}
  \Big \langle \sum_{j=1}^N e^{i k \lambda_j} \Big \rangle^{(G)} = \int_{-\infty}^\infty e^{i k \lambda}
  \rho_{(1)}(\lambda) \, d \lambda = \int_{-\infty}^\infty e^{i k \lambda}
   K^{(G)}(\lambda, \lambda)  \, d \lambda,
 \end{equation} 
 where the second equality follows from   (\ref{kp}) with $k=1$.
 
 We know from the identity (\ref{bhX}) of Br\'ezin and Hikami \cite{BH97} that
 the double integral involving the kernel $ K^{(G)}$ in (\ref{kp1}) can be reduced to a single integral
 involving $K^{(L)}|_{a=0}$. For the average in (\ref{kp2}), a result of Ullah from 1985 \cite{Ul85} (see also \cite[Appendix A]{DG01},
 \cite{HT03}, \cite{BN07}, \cite{vZ12}, \cite{CMS17}, \cite{Ok19}, \cite{Fo19}) gives
  \begin{equation}\label{kp3}
  \int_{-\infty}^\infty e^{i k \lambda} \rho_{(1)}^{(G)}(\lambda) \, d \lambda =
   e^{-k^2/4} L_{N-1}^{(1)}(k^2/2),
   \end{equation} 
   where $L^{(p)}_n(x)$ denotes the Laguerre polynomial,
   which when substituted in (\ref{kp2}) tells us
     \begin{equation}\label{kp3+}
    \Big \langle \sum_{j=1}^N e^{i k \lambda_j} \Big \rangle^{(G)} =     e^{-k^2/4} L_{N-1}^{(1)}(k^2/2).
  \end{equation}   
   
   Our aim is to give a derivation of both (\ref{kp3+}) and (\ref{bhX}) that highlights their relationship to
   the differential identities (\ref{A1a}) and (\ref{A1}). 
   For this
   we take inspiration from the working used in the recent work \cite{Ok19} to deduce these results, which although written in a quantum mechanical
   notation can be recast to serve our purpose.
    The first step in the derivation does not require use of the differential identities, rather the key inputs
  are the generating function identities for the Hermite and Laguerre polynomials
  \begin{align}
  e^{2 x t - t^2} & = \sum_{p=0}^\infty {H_p(x) t^p \over p!} \label{H1} \\
  (1 + t)^\alpha e^{xt} & = \sum_{n=0}^\infty L_n^{\alpha - n} (-x) t^n.\label{L1}
  \end{align}
  
  \begin{proposition} \label{p10} (\cite[Eq.~(2.15)]{HT03}, \cite[Eq.~(2.17)]{Ok18}, \cite[Eq.~(A.5)]{Ok19},\cite[Eq.~(C10)]{CMS17},
  \cite[Eq.~(72)]{CMC19})
  Define
  \begin{equation}\label{IG}
  I^{(G)}_{p,q}(z) = \int_{-\infty}^\infty e^{xz} H_p(x) H_q(x) e^{- x^2} \, dx.
  \end{equation} 
  We have
   \begin{equation}\label{IGa}
   I^{(G)}_{p,q}(z) = \sqrt{\pi} 2^q q! e^{(z/2)^2} z^{p-q} L_q^{p-q}(-z^2/2).
 \end{equation} 
      \end{proposition}
      
      \begin{proof}
      Using the generating function (\ref{H1}) we see
      $$
      \sum_{p,q=0}^\infty {t^p s^q \over p! q!}  I^{(G)}_{p,q}(z)  = e^{-t^2 - s^2} \int_{-\infty}^\infty e^{xz} e^{2x(t+s)} e^{-x^2} \, dx
      = \sqrt{\pi} e^{(z/2)^2 + z s} e^{t (2s+z)},
      $$
      where the second equality follows by completing the square. Equating coefficients of $t^p/p!$ on both sides gives
      (after also replacing $s$ by $s/2$)
      $$
      \sum_{q=0}^\infty {(s/2)^q \over q!} \tilde{I}_{p,q}(z) = \sqrt{\pi} e^{(z/2)^2 + z s/2}(s+z)^p.
      $$
      Making use of the generating function (\ref{L1}) on the RHS, then equating coefficients of $s^q$ on both sides,
      gives (\ref{IGa}).
      \end{proof}
      
      \begin{remark}
      Taking the coefficient of $k^{2p}$ on both sides of (\ref{kp3+}) tells us that
        \begin{equation}\label{kp3s}
        \langle k^{2p} \rangle^{(G)} = {\Gamma(1/2+p) \over \Gamma(1/2)} \sum_{l=0}^p \binom{p}{p-l} \binom{N}{l+1} 2^l =
        N {\Gamma(1/2+p) \over \Gamma(1/2)} \,  {}_2 F_1 ( - p, 1 - N; 2 ; 2),
        \end{equation}
        where $ {}_2 F_1$ denotes the Gauss hypergeometric function.
        Notice that the latter form is well defined for general complex $p$; see \cite{CMOS19} for a development of this
        point.
        \end{remark}
  
  We will now combine the differential identity (\ref{A1}) with (a special case of) Proposition \ref{p10} to derive (\ref{kp3+}).
  We have
  $$
  \int_{-\infty}^\infty e^{ik \lambda} K_N^{(G)}(\lambda,\lambda) \, d \lambda =
  {1 \over i k}  \int_{-\infty}^\infty  \Big ( {d \over d \lambda}  e^{ik \lambda} \Big ) K_N^{(G)}(\lambda,\lambda) \, d \lambda =
  -  {1 \over i k}  \int_{-\infty}^\infty  e^{ik \lambda}  {d \over d \lambda}   K_N^{(G)}(\lambda,\lambda) \, d \lambda,
  $$
 where the second equality follows by integration by parts. Next substituting for the derivative using
  (\ref{A1}) shows
  $$
 \int_{-\infty}^\infty e^{ik \lambda} K_N^{(G)}(\lambda,\lambda) \, d \lambda =   {2 \over (p_{N-1},p_{N-1})_2^{(G)}}
 {1 \over i k} \int_{-\infty}^\infty  e^{ik \lambda} \psi_N^{(G)}(\lambda)  \psi_{N-1}^{(G)}(\lambda) \, d \lambda.
 $$
 Recalling the definition of $ \psi_N^{(G)}$ from the paragraph including (\ref{GL}) we see that the integral is
proportional to the case $p=N$, $q=N-1$, $z = ik$ of (\ref{IGa}). Substituting its value, and the value too of the
inner product from (\ref{kp1y}), we obtain (\ref{kp3+}).

 To derive (\ref{bh}), rather than  combining the differential identity (\ref{A1}) with Proposition \ref{p10},
interpreting the workings in \cite{Ok19} we combine the differential identity (\ref{A1a}) with Proposition \ref{p10}.
We begin by noting
\begin{multline}
(z_1 + z_2) \Big ( {\partial \over \partial z_1} - {\partial \over \partial z_2} \Big )
\int_{\mathbb R^2} e^{x z_1 + y z_2} \Big ( 
  K_N^{(G)}(x,y) \Big )^2 \, dx d y \\
  = \int_{\mathbb R^2} (x - y) \bigg (   \Big ( {\partial \over \partial x} + {\partial \over \partial y} \Big )
  e^{x z_1 + y z_2} \bigg ) \Big ( 
  K_N^{(G)}(x,y) \Big )^2 \, dx d y .
  \end{multline}
  Integration by parts shows the RHS is equal to
  \begin{align}\label{3.18}
  - 2 & \int_{\mathbb R^2}  (x - y) e^{x z_1 + y z_2}    K_N^{(G)}(x,y)
   \Big ( {\partial \over \partial x} + {\partial \over \partial y} \Big )
  K_N^{(G)}(x,y) \, dx d y\\
& = {2 \over ((p_{N-1},p_{N-1})^{(G)})^2}      \int_{\mathbb R^2}  e^{x z_1 + y z_2} 
 \Big ( ( \psi_N^{(G)}(x)  \psi_{N-1}^{(G)}(y) )^2 - ( \psi_N^{(G)}(y)  \psi_{N-1}^{(G)}(x) )^2 \Big ) \, dx dy \nonumber \\
 & = N e^{z_1^2/4 + z_2^2/4} \Big ( L_N^{(0)}(-z_1^2/2)  L_{N-1}^{(0)}(-z_2^2/2) -  L_N^{(0)}(-z_2^2/2)  L_{N-1}^{(0)}(-z_1^2/2)
 \Big ) \nonumber \\
 & = - {1  \over (p_{N-1},p_{N-1})^{(L)}} \Big (  \psi_N^{(L)}(-z_1^2/2)  \psi_{N-1}^{(L)}(-z_2^2/2) -
  \psi_N^{(L)}(-z_2^2/2)  \psi_{N-1}^{(L)}(-z_1^2/2) \Big ) \Big |_{a=0},  \nonumber 
 \end{align} 
 where the first equality follows from the differential identity (\ref{A1a}) and the form of $K_N^{(G)}$ (\ref{kp1x});
 the second equality from Proposition \ref{p10} and (\ref{kp1y}); the third equality from the definitions in the
 paragraph including (\ref{kp1z}).
 
 Replacing $z_1 = it_1$, $z_2 = -it_2$ and recalling again (\ref{kp1x}) and (\ref{kp1z}), we have from the above
 working the differential identity \cite[Eq.~(A.18), in an equivalent form]{Ok19}
 \begin{equation}\label{3.19}
 (t_1 - t_2) \Big (  {\partial \over \partial t_1} + {\partial \over \partial t_2} \Big )
  \int_{\mathbb R^2}  e^{i t_1 x  - i t_2 y}  \Big ( K_N^{(G)}(x,y) \Big )^2 \, dx dy =
-  {t_1^2 - t_2^2 \over 2} K_N^{(L)}(t_1^2/2, t_2^2/2) \Big |_{a=0}.
\end{equation}
Taking the limit $t_2 \to t_1 = t$ gives
 \begin{equation}
 {d \over dt}  \int_{\mathbb R^2}  e^{i t (x - y)}   \Big ( K_N^{(G)}(x,y) \Big )^2 \, dx dy =
 - t  K_N^{(L)}(t^2/2, t^2/2) \Big |_{a=0},
 \end{equation}
 which upon integrating gives
 \begin{equation}\label{bh}
  \int_{\mathbb R^2} e^{i k (x - y)} \Big ( K_N^{(G)}(x,y) \Big )^2 \, dx dy =
  \int_k^\infty t K_N^{(L)}(t^2/2,t^2/2) \Big |_{a=0} \, dt.
  \end{equation}
 Substituting this in (\ref{kp1}), and also its special case $k=0$ upon noting the LHS can be
 replaced by $N$, we reclaim
 the identity (\ref{bhX}) of Br\'ezin and Hikami \cite{BH97}.
 In fact (\ref{3.19}) implies a generalisation of (\ref{bh}) as given by Proposition \ref{p11}.
 
 \medskip
  
  \noindent \emph{Proof of Proposition \ref{p11}.}
  We first cancel a factor of $t_1 - t_2$ from both sides of (\ref{3.19}). Replacing
  $t_1, t_2$ by $k_1 + s$, $k_2 + s$ respectively and using the notation
  (\ref{H}) we see that    (\ref{3.19}) is equivalent to the differential identity
  \begin{equation}\label{H2}  
  {d \over d s}    \int_{\mathbb R^2} e^{i (k_1+s)x - i(k_2 + s)y} \Big ( K_N^{(G)}(x,y) \Big )^2 \, dx dy = - H^{(L)}(k_1+s,k_2+s).
   \end{equation}
   Integrating both sides with respect to $s$ from $0$ to $\infty$ then gives 
 \begin{equation}\label{H1}   
     \int_{\mathbb R^2} e^{i k_1x - ik_2y} \Big ( K_N^{(G)}(x,y) \Big )^2 \, dx dy =
    \int_0^\infty H^{(L)}(k_1+s, k_2 + s) \, ds.
  \end{equation} 
  Note that (\ref{H1}) reduces to   (\ref{bh}) in the special case $k_1 = k_2 = k$.
  
  In general, with $A = \sum_{j=1}^N a(\lambda_j)$, $B = \sum_{j=1}^N b(\lambda_j)$,
  we have in accordance with (\ref{3}) and (\ref{Cov})
  $$
  {\rm Cov} \, (A,B) = \int_I d \lambda \int_I d \lambda' \, a(\lambda) b(\lambda) N_{(2)}(\lambda, \lambda').
  $$
  Recalling (\ref{1})--(\ref{2a}), and (\ref{kp}) with $k=1$ and 2, we see that for a determinantal point process
  $$
   {\rm Cov} \, (A,B) =   {1 \over 2} \int_I d \lambda \int_I d \lambda' \, \Big (a(\lambda) - a(\lambda') \Big )
 \Big (b(\lambda) - b(\lambda') \Big ) \Big ( K_N(\lambda, \lambda') \Big )^2.
 $$
 For the particular choices $A = \sum_{j=1}^N e^{ik_1 \lambda_j}$, $B = \sum_{j=1}^N e^{-ik_2 \lambda_j}$,
 and under the assumption $K(\lambda,\lambda') = K(\lambda',\lambda)$ as is valid for the GUE,
 this implies
  \begin{equation}\label{3.23a}
  {\rm Cov} \, (A,B) = \int_I d \lambda \int_I d \lambda' \, \Big ( e^{i (k_1 - k_2) \lambda} - 
  e^{ik_1 \lambda}   e^{-ik_2 \lambda'} \Big )  \Big ( K_N(\lambda, \lambda') \Big )^2.
  \end{equation}
  Writing the RHS as the difference of two integrals, we see that for the GUE, both can
  be evaluated using (\ref{H1}), and the sought identity (\ref{H1X}) results.
  \hfill $\square$
  
  \begin{remark} 1. In the particular case $t_1 = i \Gamma +k, t_2 = - i \Gamma + k$
  we see that (\ref{3.19}) reduces to
  \begin{equation}\label{gs}
  {\partial \over \partial k} \int_{\mathbb R^2} e^{- \Gamma (x + y)} e^{i k (x-y)} \Big ( K_N^{(G)}(x,y) \Big )^2 \, dx dy = - k
  K_N^{(L)}(t_1^2/2, t_2^2/2).
  \end{equation}
  This is an equivalent form of \cite[Eq.~(2.19)]{Ok19}. \\
  2.
  Setting $k_2=0$ in (\ref{H1}) allows the LHS to be evaluated
  $$
  \int_{\mathbb R^2} e^{i k_1 x} \Big ( K_N^{(G)}(x,y) \Big )^2 \, dx dy =
  \int_{-\infty}^\infty e^{i k_1 x} K_N^{(G)}(x,x) \, dx = e^{-k_1^2/4} L_N^{(1)}(k_1^2/2),
  $$
  where the first equality follows by integrating over $y$ using the form of $K_N^{(G)}(x,y)$ given
  by the first equality in (\ref{kp1x}), and the second equality follows from (\ref{kp3}). Hence
  $$
  \int_0^\infty H^{(L)}(k_1+s,s) \, ds = e^{- k_1^2/4} L_N^{(1)}(k_1^2/2),
  $$
  which seems difficult to deduce directly from the definition (\ref{H}) of $H^{(L)}$. 
  \end{remark}
  
  \subsection{Global scaling limits with $k = \tau_g/\sqrt{2N}$}
  With the scaling $k = \tau_g/\sqrt{2N}$ we have $e^{i k \lambda_j} = e^{i \tau_g x_j}$, where
  $x_j = \lambda_j/\sqrt{2N}$. From the discussion in the paragraph including (\ref{8a}), in
  terms of the scaled eigenvalues $\{x_j\}$, the spectrum has leading order support $|x|<1$
  and thus $e^{i \tau_g x_j}$ slowly varies on the scale of the global spectrum.
  
  In keeping with the general limit formula (\ref{8a}) for the variance of a linear statistic, in this 
  setting there are limit formulas associated with the averages and corresponding evaluations given
  in (\ref{kp3+}), (\ref{bhX}) and (\ref{H1X}).
  
  \begin{proposition}
  Let $J_\nu(x)$ denote the Bessel function. We have
   \begin{align}
   \lim_{N \to \infty} {1 \over N}    \Big \langle \sum_{j=1}^N e^{i \tau_g \lambda_j / \sqrt{2N}} \Big \rangle^{(G)} & =  { J_1(\tau_g) \over \tau_g} , \label{D1} \\
   \lim_{N \to \infty} \bar{S}_N^{(G)}(\tau_g/\sqrt{2N}) & = {\tau_g^2 \over 2} \Big ( (J_0(\tau_g))^2 + (J_1(\tau_g))^2  - {J_0(\tau_g) J_1(\tau_g) \over \tau_g} \Big ), \label{D2} 
   \end{align}
    and
      \begin{multline}
   \lim_{N \to \infty} {\rm Cov} \,   \Big ( \sum_{j=1}^N e^{i \tau_g^{(1)} \lambda_j / \sqrt{2N}} ,   \sum_{j=1}^N e^{i \tau_g^{(2)} \lambda_j / \sqrt{2N}}  \Big )^{(G)} \\ =
   - {\tau_g^{(1)}    \tau_g^{(2)} \over 2 ( \tau_g^{(1)} + \tau_g^{(2)})} \Big ( J_0(\tau_g^{(1)}) J_1(\tau_g^{(2)}) + J_1(\tau_g^{(1)}) J_0(\tau_g^{(2)}) \Big ).\label{D3}
\end{multline}
  \end{proposition}
  
  \begin{proof}
  It follows from (\ref{kp3+}), upon use of the explicit form of the Laguerre polynomial, that
   \begin{equation}\label{D1+}   
  {1 \over N}    \Big \langle \sum_{j=1}^N e^{i \tau_g \lambda_j / \sqrt{2N}} \Big \rangle^{(G)} = e^{- \tau_g^2/8N} \sum_{p=0}^{N-1} {(-N+1)_p \over p! (p+1)!} \Big ( {\tau_g^2 \over 4N} \Big )^p,
  \end{equation}
  where $(u)_p =u (u+1)\cdots(u+p-1)$ denotes the Pochhammer symbol for the increasing factorial. Noting 
  $$
  (-N+1)_p {1 \over N^p} = (-1)^p \bigg ( 1 - {p(p+1) \over 2 N} + {\rm O} \Big ( {1 \over N^2} \Big ) \bigg ),
  $$
  and expanding the exponential to ${\rm O}(1/N^2)$, we see that for large $N$ (\ref{D1+}) has the form
   \begin{equation}\label{D1+a}  
  \sum_{p=0}^\infty {1 \over p! (p+1)!} \Big ( {\tau_g \over 2} \Big )^{2p} + {\rm O} \Big ( {1 \over N^2} \Big ).
    \end{equation}
   The limit formula (\ref{D1}) follows upon recalling the series form of the Bessel function.
    
    In  relation to (\ref{D2}), a simple change of variables shows
    \begin{equation}\label{D2+}   
      \lim_{N \to \infty} \bar{S}_N^{(G)}(\tau_g/\sqrt{2N}) = 2 \lim_{N \to \infty} {1 \over 4N} \int_0^{\tau_g} t K_N^{(L)} \Big ( {t^2 \over 4N}, {t^2 \over 4N} \Big ) \Big |_{a=0} dt. 
      \end{equation}
      We know \cite{FT19} that for large $N$
        \begin{equation}\label{D2a}  
      {1 \over 4N}     K_N^{(L)} \Big ( {t^2 \over 4N}, {t^2 \over 4N} \Big ) = \rho_{(1)}^{\rm hard}(t^2) \Big |_{a=0} +  {\rm O} \Big ( {1 \over N^2} \Big ),
     \end{equation}
     where the error term holds uniformly in a compact set, and where \cite{Fo93a}
      \begin{equation}\label{D2a+}  
    \rho_{(1)}^{\rm hard}(t^2) \Big |_{a=0} = (J_0(t))^2 +  (J_1(t))^2
    \end{equation}
    is the so-called hard edge scaled density  in the variable $t^2$, Laguerre parameter $a=0$.
     Substituting in (\ref{D2+}) and evaluating the integral (we used computer algebra) gives (\ref{D2}).
    
    It remains to consider (\ref{D3}). Following the working of the above paragraph, the essential point is the fact that
    for large $N$  \cite{FT19} 
        \begin{equation}\label{D3a}  
      {1 \over 4N}     K_N^{(L)} \Big ( {(\tau_g^{(1)} - \tau_g^{(2)} + t)^2 \over 4N}, {t^2 \over 4N} \Big ) = K^{\rm hard}((\tau_g^{(1)} - \tau_g^{(2)} + t)^2, t^2) \Big |_{a=0} +  {\rm O} \Big ( {1 \over N^2} \Big ),
     \end{equation}
     where 
    $$
   K^{\rm hard} (X^2,Y^2)     \Big |_{a=0}  =  {X J_1(X) J_0(Y) - Y J_1(Y) J_0(X) \over X^2 - Y^2}
   $$
    is the corresponding hard edge scaled correlation kernel \cite{Fo93a}. Hence, we see from  (\ref{H1X}) that
         \begin{multline}\label{D3a+} 
   \lim_{N \to \infty} {\rm Cov} \,   \Big ( \sum_{j=1}^N e^{i \tau_g^{(1)} \lambda_j / \sqrt{2N}} ,   \sum_{j=1}^N e^{i \tau_g^{(2)} \lambda_j / \sqrt{2N}}  \Big )^{(G)} \\ =
   {1 \over \tau_g^{(1)} - \tau_g^{(2)} } \int_0^{\tau_g^{(2)}}  \Big ( ( \tau_g^{(1)} - \tau_g^{(2)} + t) J_1(  \tau_g^{(1)} - \tau_g^{(2)} + t)  J_0(t) - t J_1(t)  J_0(  \tau_g^{(1)} - \tau_g^{(2)} + t) \Big ) \, dt.
 \end{multline}
 We can check that for $a$ fixed\footnote{The package Mathematica can verify this, but was unable to integrate the LHS independently.}
 $$
 (a + t) J_1(a+t) J_0(t) - t J_1(t) J_0(a+t) = - {d \over dt}   \bigg (  (a+t) t \Big ( J_0(a+t) J_1(t) - J_0(t) J_1(a+t) \Big ) \bigg ).
 $$
 Substituting in (\ref{D3a+}) allows the integral to be computed, and (\ref{D3}) results.
 \end{proof}
 
 \begin{remark} 1. Setting $\tau_g^{(1)}=\tau$ and taking the limit $\tau_g^{(1)} \to - \tau$ we see that
 (\ref{D3}) reduces to (\ref{D2}) in keeping with the definition of $\bar{S}^{(G)}$ as it relates to the covariance. 
 The formula (\ref{D3}) itself, with $\tau_g^{(1)}  = - i k_1, \tau_g^{(1)}  = i k_2$ has appeared in earlier work
 \cite{AD02,GPR10,CFM19}, computed directly from (\ref{8a}) as modified to apply to the covariance; we owe
 \cite[Eqns.~(3.9) and (3.10), and surrounding text]{Ok18} for knowledge of this.
 \\
 2. From the asymptotic expansions used in the above proof, we see that the rate of convergence is ${\rm O}(1/N^2)$.
 This is in keeping with the well known fact that the loop equation for the GUE, which relate to connected
 correlations for the linear statistic $\sum_{j=1}^N {1 \over z - \lambda_j}$ in the global scaling limit,
 permit an expansion in powers of $1/N^2$; see e.g.~\cite{WF14}. \\
 3. From the definitions, and upon a simple change of variables
     \begin{equation}\label{Wa}
 {1 \over N}   \Big \langle \sum_{j=1}^N e^{i \tau_g \lambda_j / \sqrt{2N}} \Big \rangle^{(G)} =
 {1 \over N} \sqrt{2N} \int_{-\infty}^\infty e^{i \tau_g x} \rho_{(1)}^{\rm GUE}(\sqrt{2N} x) \, dx.
  \end{equation}
 The density on the RHS is the global scaled spectral density of the GUE, which limits to the
 so-called Wigner semi-circle law (see e.g.~\cite{PS11})
    \begin{equation}\label{W}
    \lim_{N \to \infty}  {1 \over N} \sqrt{2N}   \rho_{(1)}^{\rm GUE}(\sqrt{2N} x)  = \rho^{\rm W}(x), \qquad
 \rho^{\rm W}(x) = {2 \over \pi} (1 - x^2)^{1/2} \chi_{|x| < 1},
 \end{equation}
 where $\chi_A = 1$ for $A$ true, $\chi_A = 0$ otherwise. The result (\ref{D1}) can therefore be rewritten
 $$
 \int_{-1}^1 e^{i \tau_g x} \rho^{\rm W}(x) \, dx =  { J_1(\tau_g) \over \tau_g},
 $$
 which is well known as the exponential generating function
  for the moments of the Wigner semi-circle law. \\
  4. Denote $\lim_{N \to \infty} \bar{S}_N^{(G)}(\tau_g/\sqrt{2N}) = \bar{S}_\infty^{(G)} (\tau_g)$. It follows
  from (\ref{D2}) that for $\tau_g \to \infty$
  $$
   \bar{S}_\infty^{(G)} (\tau_g) = {\tau_g \over \pi} + {\rm O}(1).
   $$
   This corresponds to the linear (referred to as a ramp in e.g.~ \cite{CHLY17}) behaviour exhibited in 
   (\ref{7}) and  (\ref{9}).  
 \end{remark}
 
 \subsection{Bulk scaling limits with $k = 2 \sqrt{2N} \tau_b$} 
 With this scaling we have $e^{i k \lambda_j} = e^{i \tau_b x_j}$, where
 $x_j = 2 \sqrt{2N} \lambda_j$. To appreciate the significance of the latter, we recall
 (see e.g.~\cite[\S 7.1.1]{Fo10}) that for the GUE in the bulk of the spectrum the mean
 spacing between eigenvalues is ${\rm O}(1/\sqrt{N})$, so the mean spacing with respect
 to the scaled eigenvalues $\{ x_j \}$ is ${\rm O}(1)$ thereby corresponding to bulk scaling ---
 recall too the paragraph including (\ref{4}).
 
 Upon this scaling of $k$, the average (\ref{kp2}) does not exhibit a scaling form for large $N$. 
 Instead we see from (\ref{Wa}) and (\ref{W}) that for large $N,\tau_b$
  \begin{equation}\label{Wb}
  \Big \langle \sum_{j=1}^N e^{i  2 \sqrt{2N} \tau_b \lambda_j} \Big \rangle^{(G)} \sim {2 N \over \pi}
  \int_{-1}^1 e^{i 4 N \tau_b x} (1 - x^2)^{1/2} \, dx \sim {1 \over 2 \sqrt{2 \pi N} \tau_b^{3/2}} \cos (4 N \tau_b - 3 \pi/4),
  \end{equation}
 where the second asymptotic form follows by expanding in the neighbourhood of $\pm 1$ to leading order
 and changing variables. This decay of order $\tau_b^{-3/2}$ is responsible for the dip effect
 in the shape of the graph of the first average on the RHS of (\ref{8b}), as observed in \cite{C+17}.
 The modification of (\ref{Wb}) valid for $0 < \tau_b < 1$ is given in (\ref{g1}) below.
 
 In contrast, there is a scaled functional form associated with the bulk scaling of $\bar{S}^{(G)}(k)$.
 As noted in the original paper of Br\'ezin and Hikami \cite{BH97}, the formula (\ref{bhX}) allows for
  an easy derivation of its explicit form, which
  relates to the CUE result (\ref{7}). Furthermore, the recent work of Okuyama \cite{Ok19} --- already crucial in
  our derivation of Proposition \ref{p11} --- reveals the
  appropriate bulk scaling generalisation of (\ref{D3}).
  
  \begin{proposition} 
  Define
  \begin{equation}\label{g}
 f_{N}^{(\rm L)}(X) = \sin \Big ( 2N ( \sqrt{X(1-X)} + {\rm Arcsin}\, \sqrt{X} ) 
-  \pi / 4 \Big ).
     \end{equation}
  For large $N$, and with $0 < \tau_b < 1$, we have
  \begin{equation}\label{g1}     
  \Big \langle \sum_{j=1}^N e^{i 2 \sqrt{2N} \tau_b \lambda_j } \Big \rangle^{(G)} =
  (4 N \tau_b^2)^{-1/2} \Big ( 2 \pi \tau_b \sqrt{1 - \tau_b^2} \Big )^{-1/2} \bigg ( f_{N}(\tau_b^2) +
  {\rm O} \Big ( {1 \over N} \Big ) \bigg ).  
   \end{equation}   
 Furthermore \cite{BH97}
  \begin{equation}\label{MPa}
  \lim_{N \to \infty} {1 \over N} \bar{S}_N^{(G)}(2 \sqrt{2N} \tau_b ) = 
  \left \{ \begin{array}{ll} {2 \over \pi} (\tau_b \sqrt{(1 - \tau_b^2)} + {\rm Arcsin} \, \tau_b ), & 0 < \tau_b  < 1, \\
  1, & \tau_b > 1,
  \end{array} \right.
   \end{equation}
   and with $\gamma \in \mathbb R$ \cite[Eq.(3.16)]{Ok19}
    \begin{multline}\label{MPb}
   \lim_{N \to \infty} {1 \over N}   {\rm Cov } \, \Big ( \sum_{j=1}^N e^{ i k_1 \lambda_j},   \sum_{j=1}^N e^{ - i k_2  \lambda_j} \Big )^{(G)}
   \bigg |_{k_1 = i \gamma/ (2 \sqrt{2N}) + 2 \sqrt{2N} \tau \atop
   k_2 = - i \gamma/ (2 \sqrt{2N}) + 2 \sqrt{2N} \tau} \\
   = \left \{ \begin{array}{ll} {4 \over \gamma \pi} \int_0^\tau \sinh ( \gamma (1 - s^2)^{1/2}) \, ds & 0 < \tau  < 1, \\[.2mm]
  1, & \tau > 1.
   \end{array} \right.
   \end{multline}
   \end{proposition}
   
   \begin{proof}
   The large $N$ form (\ref{g1}) follows from (\ref{kp3+}) upon applying a version of the
   Plancherel-Rotach asymptotic formula for the Laguerre polynomials derived in
   \cite[minor rewrite of eqns.~(3.12) \& (3.13)]{FFG06},
   \begin{multline}\label{7.160p}
 x^{a/2} e^{-x/2} L_{n+m}^{a}(x) |_{x = 4n X} =
(-1)^{m} (2 \pi \sqrt{X(1-X)})^{-1/2} n^{a/2 - 1/2} \Big (
g_{m,n}^{(L)}(X) + {\rm O} \Big ( {1 \over n} \Big ) \Big ),\\
 g_{m,n}^{(L)}(X) := \sin \Big ( 2n ( \sqrt{X(1-X)} + {\rm Arcsin}\, \sqrt{X} ) -
(2m + a + 1) {\rm Arccos}\, \sqrt{X} + 3 \pi / 4 \Big ),
\end{multline}
where we set $a=1, n=N, m = -1$, and it is assumed $0 < X < 1$.
  
   In relation to (\ref{MPa}), after changing variables $t^2/2 = s$, it follows from (\ref{bhX}) that
   $$
  \lim_{N \to \infty} {1 \over N} \bar{S}_N^{(G)}(2 \sqrt{2N} \tau_b ) = \lim_{N \to \infty} \int_0^{\tau_b^2} 4  \rho_{(1)}^{\rm LUE}(4Ns) \Big |_{a=0} \, ds.
  $$
  But with $ \rho^{\rm MP}(x)$ denoting the particular Mar\u{c}enko--Pastur density 
  $$
  \rho^{\rm MP}(x) = {2 \over \pi x^{1/2}} (1 - x)^{1/2} \chi_{0 < x < 1},
 $$
  we have
  \cite{PS11}
  $$
  \lim_{N \to \infty} 4  \rho_{(1)}^{\rm LUE}(4Ns) \Big |_{a=0}  =  \rho^{\rm MP}(s).
  $$
  Substituting and evaluating the integral gives  (\ref{MPa}).
  
  It remains to consider (\ref{MPb}). According to (\ref{H1X}), after changing variables $s \mapsto s - i \gamma/(2\sqrt{2N})$,
  then $s \mapsto 2 \sqrt{2N}$ in the integral therein,
  \begin{multline}\label{t.1}
  {1 \over N}   {\rm Cov } \, \Big ( \sum_{j=1}^N e^{ i k_1 \lambda_j},   \sum_{j=1}^N e^{ - i k_2  \lambda_j} \Big )^{(G)}
   \bigg |_{k_1 = i \gamma/ (2 \sqrt{2N}) + 2 \sqrt{2N} \tau \atop
   k_2 = - i \gamma/ (2 \sqrt{2N}) + 2 \sqrt{2N} \tau} \\
   = {2 \sqrt{2N} \over N} \int_0^\tau  H^{(L)}(2 \sqrt{2N}(s + i \gamma/8N, 2 \sqrt{2N}(s - i \gamma/8N)) \, ds.
      \end{multline}
  Set
  $$
  t_1 = 2 \sqrt{2N}(s + i \gamma/8N), \qquad  t_2 = 2 \sqrt{2N}(s - i \gamma/8N).
  $$
  We see from the final equality in (\ref{3.18}) and the definition (\ref{H}) that
    \begin{equation}\label{t.2} 
    H^{(L)}(t_1,t_2) = - {(2N)^{3/2} \over \gamma} e^{- t_1^2/4 - t_2^2/4} {\rm Im} \, \Big ( L_N^{(0)}(t_1^2/2)
     L_{N-1}^{(0)}(t_2^2/2) \Big ).
     \end{equation}
     
     Consider first the range  $0 < \tau < 1$. 
     Making use of the Plancherel-Rotach asymptotic formula (\ref{7.160p}) with $a=0, n = N, m=0,-1$,
     expanding the term proportional to $N$ in the cosine there according to
  \begin{multline*}    
  \Big ( \sqrt{X ( 1 - X)} + {\rm Arcsin} \, \sqrt{X} \Big ) \Big |_{X = s^2 + i s \gamma/4N} \\
  = s \sqrt{1 - s^2} + {\rm Arcsin} \, s + {i \gamma \over 4 N} (1 - s^2)^{1/2} + {\rm O} \Big ( {1 \over N^2} \Big ),
  \end{multline*}
  and making use too of a simple trigonometric identity shows
     \begin{multline}\label{t.3} 
     - {\rm Im} \, \Big ( L_N^{(0)}(t_1^2/2)
     L_{N-1}^{(0)}(t_2^2/2) \Big ) \\
     = {1 \over 4 \pi N s \sqrt{1 - s^2}}
     \bigg ( \cos \Big ( i \gamma (1 - s^2)^{1/2} - 2 {\rm Arccos} \, s \Big ) - \cos ( N v(s)) +  {\rm O} \Big ( {1 \over N} \Big ) \bigg ).
     \end{multline}
     The explicit form of the function $v(s)$ follows from the above working, however the only property we require is that it is
     bounded and linear in $s$ for small $s$. Substituting (\ref{t.3}) in (\ref{t.2}), then substituting the result in (\ref{t.1}), from
     the aforementioned property of $v(s)$, we see that the term involving $\cos N v(s)$ contributes ${\rm O}(1/N)$ to the integral
     relative to the other trigonometric term in (\ref{t.3}). Thus for large $N$ the LHS of (\ref{t.1}) equals
 \begin{equation}\label{t.4}   
 {2 \over \pi \gamma s \sqrt{1 - s^2}} {\rm Im} \, \int_0^\tau \cos \Big (    ( i \gamma (1 - s^2)^{1/2} - 2 {\rm Arccos} \, s \Big ) \, ds
 + {\rm O} (N^{-1}).
 \end{equation}
 Simplification gives the functional form in (\ref{MPb}). 
 
 The result (\ref{MPb}) for the range $\tau > 1$ follows from the result for $\tau = 1$, since for $s>1$ the appropriate version of
 the Plancherel-Rotach asymptotic formula \cite{Mo34} tells us that (\ref{t.3}) is exponentially small in $N$.

  \end{proof}
  
  \begin{remark}
 1.  The asymptotic formula of the final sentence in the above proof tells us that 
  the average in (\ref{g1}) for $\tau_g > 1$ tends to
  zero exponentially fast in $N$. Similarly,
  for $s > 1$ we know \cite{Fo12} that $ \rho_{(1)}^{\rm LUE}(4Ns)$ is exponentially
  small in $N$, and the limiting value of unity in (\ref{MPa}) is approached exponentially fast.
  The explicit form of the correction terms in this context are discussed in \cite{Ok19a}.
  In contrast,
 from G\"otze and Tikhomirov \cite{GT05}, we know that for any fixed value of the Laguerre parameter
 $a$,
 $$
 \sup_{x} \Big | 4 \int_0^x  \rho_{(1)}^{\rm LUE}(4Ns) ds - \int_0^x \rho^{\rm MP}(u) \, du \Big |
 \le {C \over N}
 $$
 for some $C > 0$, independent of $N$, which is furthermore optimal. Hence the limiting value in
 (\ref{MPa}) for $0 < \tau < 1$ is approached at a rate ${\rm O}(1/N)$. \\
 2. The working in \cite{Ok19} leading to (\ref{MPb}) did not make use of (\ref{gs}) --- although derived
 in the same paper --- but rather proceeded by extending heuristic working based on the double
 integral form the covariance (\ref{3.23a}), used previously in \cite{BH97,Li18}.
 \end{remark}
 
 \subsection{Soft edge scaling limits $k = \sqrt{2} N^{1/2} i \gamma$}
 The neighbourhood of the largest eigenvalue in the GUE gives rise to a well defined
 determinantal point process upon the change of variables $\lambda_j = \sqrt{2N} + x_j/(\sqrt{2} N^{1/6})$,
 and taking the limit $N \to \infty$.
 Motivating this choice are the facts that the largest eigenvalue to leading order is equal to
 $\sqrt{2N}$, with neighbouring eigenvalues separated on a scale of order $1/N^{1/6}$ \cite{Fo93a}.
 In the variables $\{x_j\}$, the eigenvalue density increases like $\sqrt{|x|}/\pi$ as $x \to -\infty$
 (see \cite[Eq.~(7.69)]{Fo10})
 which is consistent with the functional form of the Wigner semi-circle density (\ref{W}) at the edge,
 while there is a decay of leading order $\exp ( - 4 x^{3/2}/3 )$ for $x \to \infty$. This latter feature gives
 rise to the terminology of a soft edge. The asymptotics follow from the explicit form of the correlation
 kernel \cite{Fo93a},
 \begin{equation}\label{Ksoft}
 K^{\rm soft}(x,y) = {{\rm Ai} (x) {\rm Ai}'(y) - {\rm Ai} (y) {\rm Ai}'(x) \over x - y},
 \end{equation}
 which implies
  \begin{equation}\label{Ksoft1}
  \rho_{(1)}^{\rm soft}(x) = - x  ({\rm Ai}(x))^2 +  ({\rm Ai}'(x))^2 .
   \end{equation}
   
   Averages and covariances are defined in terms of the correlations as determined by (\ref{Ksoft}).
   In particular
   \begin{align}
   \langle e^{\gamma x} \rangle^{\rm soft} & = \int_{-\infty}^\infty e^{\gamma x} \rho_{(1)}^{\rm soft}(x) \, dx \label{Q1} \\
   {\rm Cov} \, (e^{\gamma_1 x}, e^{\gamma_2 y})^{\rm soft} & = \int_{\mathbb R^2} (e^{\gamma_1 x} - e^{\gamma_1 y})
   (e^{\gamma_2 x} - e^{\gamma_2 y}) \Big ( K^{\rm soft}(x,y) \Big )^2 \, dx dy \label{Q2}
   \end{align}
   which for convergence require that the parameters $\gamma, \gamma_1, \gamma_2$ be positive. Interest
   in the quantities (\ref{Q1}), (\ref{Q2}) first came about as generating functions relating to intersection numbers
   on the moduli space of certain families of algebraic curves \cite{Ok02}. Direct calculation based on an integral
   formula for (\ref{Ksoft})
    \begin{equation}\label{Ksoft2}
   K^{\rm soft}(x,y) =  \int_0^\infty {\rm Ai} (x +t)    {\rm Ai} (y +t) \, dt
   \end{equation}
   allows for the evaluations
     \begin{align}
     \langle  e^{\gamma x} \rangle^{\rm soft} & = { e^{\gamma^3/12} \over 2 \sqrt{\pi} \gamma^{3/2} } \label{Q3} \\
     {\rm Cov} \, (e^{\gamma_1 x}, e^{\gamma_2 y})^{\rm soft} & = {e^{(\gamma_1 + \gamma_2)^3/12} \over 2 \sqrt{\pi}
     (\gamma_1 + \gamma_2)^{3/2}} {\rm Erf} \, \Big ( {1 \over 2} \sqrt{\gamma_1 \gamma_2 (\gamma_1 + \gamma_2)} \Big ); \label{Q4}
   \end{align}
  see \cite{OS20} for a clear statement.
  
  Our point is that these can also be related to limits of the GUE average (\ref{kp3}) and covariance (\ref{H1X}). For example,
  from the eigenvalue scaling of the GUE specifying the soft edge, we must have, for $\gamma > 0$
  $$
  \langle e^{\gamma x} \rangle^{\rm soft} = \lim_{N \to \infty} e^{-ik \sqrt{2N}} 
   \int_{-\infty}^\infty e^{i k \lambda} \rho_{(1)}^{(G)}(\lambda) \, d \lambda  \bigg |_{k = - i \gamma \sqrt{2} N^{1/6}}.
   $$
   Hence, from (\ref{kp3}) and (\ref{Q3})
      \begin{equation}\label{Q5}
   { e^{\gamma^3/12} \over 2 \sqrt{\pi} \gamma^{3/2} }  = \lim_{N \to \infty} e^{- 2 \gamma N^{2/3}} e^{\gamma^2 N^{1/3}/2}
    L_{N-1}^{(1)}(-\gamma^2 N^{1/3}).
    \end{equation}
    Using an integral representation of the Laguerre polynomial, this identity (in an equivalent form) has been derived using a saddle
    point analysis by Br\'ezin and  Hikami \cite{BH07}.
    
  \subsection*{Acknowledgements}
	This research is part of the program of study supported
	by the Australian Research Council Centre of Excellence ACEMS.

\nopagebreak

\providecommand{\bysame}{\leavevmode\hbox to3em{\hrulefill}\thinspace}
\providecommand{\MR}{\relax\ifhmode\unskip\space\fi MR }
\providecommand{\MRhref}[2]{%
  \href{http://www.ams.org/mathscinet-getitem?mr=#1}{#2}
}
\providecommand{\href}[2]{#2}

   \end{document}